\documentclass{article}
\usepackage[utf8]{inputenc}
\usepackage{amsmath,amssymb,amsthm}
\usepackage{bm}
\usepackage{algorithm,algorithmic}
\usepackage{enumerate}
\usepackage[margin=25mm]{geometry}
\usepackage[numbers]{natbib}
\usepackage{mathtools}
\mathtoolsset{showonlyrefs}

\newcommand{\argmin}{\mathop{\rm argmin}\limits}

\newcommand{\proj}{\mathrm{proj}}
\newcommand{\dist}{\mathrm{dist}}

\newtheorem{theorem}{Theorem}[section]
\newtheorem{corollary}{Corollary}[section]
\newtheorem{lemma}{Lemma}[section]

\newtheorem{proposition}{Proposition}[section]
\newtheorem{assumption}{Assumption}[section]

\def\innerprod<#1>{\langle #1 \rangle}

\title{Optimal reinsurance and investment via stochastic projected gradient method based on Malliavin calculus}
\author{Yuta Otsuki\thanks{Department of Data Science for Business Innovation, Chuo University, Japan, E-mail: a19.8s78@g.chuo-u.ac.jp}
\and
Shotaro Yagishita\thanks{The Institute of Statistical Mathematics, Japan, E-mail: syagi@ism.ac.jp}}
\date{\today}

\begin{document}

\maketitle

\begin{abstract}
This paper proposes a new approach using the stochastic projected gradient method and Malliavin calculus for optimal reinsurance and investment strategies.
Unlike traditional methodologies, we aim to optimize static investment and reinsurance strategies by directly minimizing the ruin probability.
Furthermore, we provide a convergence analysis of the stochastic projected gradient method for general constrained optimization problems whose objective function has H\"older continuous gradient.
Numerical experiments show the effectiveness of our proposed method.
\end{abstract}

\section{Introduction}\label{sec:intro}
Non-life insurance companies employ reinsurance and financial asset investments as essential elements of their risk management strategy.
Reinsurance, in particular, helps insurers mitigate the risk of ruin by shielding them from large losses and ensuring solvency.
In risk theory, several notions of ruin probability are considered, such as
\begin{align}
    &\mathbb{P} \left(\inf_{t \ge 0} U_t < 0\right),  \label{eq:ruin_prob_infty}\\
    &\mathbb{P} \left(\inf_{0\le t \le T} U_t < 0\right),\label{eq:ruin_prob_finte}\\
    &\mathbb{P} \left(U_T< 0\right), \label{eq:ruin_prob}
\end{align}
where $U_t$ is the surplus process of the insurer.
In a vast literature, finding reinsurance and investment strategies that minimize the ruin probability has been considered.
Such studies generally fall into two primary approaches: the Hamilton–Jacobi–Bellman (HJB) approach and the adjustment coefficient approach.

The HJB approach deals with a stochastic optimal control problem that aims to minimize the probability of financial ruin while considering dynamic investment and reinsurance strategies.
Optimal strategies are derived from the solution of the HJB equation.
\citet{schmidli2001optimal} considered optimal proportional reinsurance by minimizing the infinite time ruin probability \eqref{eq:ruin_prob_infty} in the Cramér-Lundberg model and presented a method for solving the HJB equation to determine the optimal strategy.
\citet{schmidli2002minimizing} further investigated optimal reinsurance and investment.
While the HJB approach directly minimizes the ruin probability, they are less realistic in insurance operations because the decision variable is a continuous-time stochastic process.

The adjustment coefficient approach is a method based on the Lundberg inequality
\begin{align}\label{eq:Lun_ineq}
    \mathbb{P}\left(\inf_{t\ge0} U_t < 0\right) < e^{-R(\eta)u}
\end{align}
where $R(\eta)$ is called the adjustment coefficient and $\eta$ represents strategies, and minimizing a function on the right-hand side of \eqref{eq:Lun_ineq}, the upper bound of the ruin probability, with respect to $\eta$.
\citet{waters1983some} studied optimal proportional reinsurance and excess of loss reinsurance in the Cramér-Lundberg model.
\citet{waters1983some} proved that the adjustment coefficient is a unimodal function for proportional reinsurance and also for excess reinsurance under certain assumptions.
\citet{CENTENO1986169} investigated optimal proportional reinsurance, excess of loss reinsurance, and their combination for the Sparre Anderson model.
\citet{hald2004maximisation} derived the closed-form expressions of the optimal proportional reinsurance under the Cramér-Lundberg model and the Sparre Anderson model for the first time.
\citet{liang2007optimal} also investigated the optimal proportional reinsurance in a jump-diffusion model.
Although only reinsurance was considered in the above studies, \citet{liang2008upper} considered the combination of proportional reinsurance and investment for a jump-diffusion model.
\citet{liang2012optimal} also studied the optimal investment and proportional reinsurance for the Sparre Andersen model.
There have been subsequent studies on reinsurance and investment using the adjustment coefficient approach (see, e.g., \citep{hu2012optimal,zhang2016optimal,meng2023multiple}).
The strategies obtained through the adjustment coefficient approach are static, where the decision variables are real numbers, as opposed to dynamic strategies, such as those in the HJB approach.
Therefore, the adjustment coefficient approach is more practical than the HJB approach.
On the other hand, this approach minimizes an upper bound for the ruin probability, rather than directly minimizing the ruin probability itself.

In this paper, we attempt to minimize the ruin probability directly.
Precisely, a constrained optimization problem that minimizes \eqref{eq:ruin_prob} with respect to proportional reinsurance and investment is considered.
Our approach is more practical than the HJB approach because the strategies in ours is static as well as that in the adjustment coefficient approach.
Fortunately, the constraint is simple and using Malliavin calculus for the Poisson process \citep{privault2004malliavin} yields an unbiased estimator of the gradient of the objective function.
Consequently, we can apply a mini-batch stochastic projected gradient method to the optimization problem.
Furthermore, the convergence of the algorithm is shown by investigating the H\"older continuity of the gradient.

The contributions of this paper are given as follows:
\begin{itemize}
    \item There has been no approach to consider static strategies that directly minimize the ruin probability.
    Furthermore, to our knowledge, the adjustment coefficient approach cannot be employed in the surplus process with the geometric Brownian motion as the asset.
    In contrast, the geometric Brownian motion can be used in our approach.

    \item \citet{ewald2006malliavin} and \citet{ewald2006new} proposed a gradient-type method using the Malliavin calculus for the Brownian motion to compute the gradient of the objective function.
    Their algorithms assume that the gradient is calculated exactly and does not account for Monte Carlo errors.
    In contrast, we enhance computational efficiency by employing a stochastic gradient-based algorithm that does not compute the gradient exactly.
    In addition, while they did not include any convergence analysis, we provide convergence results by investigating the properties of gradients.

    \item We show the convergence of a mini-batch stochastic projected gradient method for constrained problems minimizing objective function with H\"older continuous gradient, not limited to the ruin probability minimization problems.
    In the literature \citep{lei2019stochastic,patel2021stochastic,patel2022global}, theoretical analyses of the stochastic gradient method under the assumption of H\"older continuity of the gradient are provided.
    Although they are results for unconstrained problems, we establish non-asymptotic analysis for constrained optimization problems.
\end{itemize}

The rest of the paper is organized as follows.
The remainder of this section is devoted to notation and preliminary results.
In the next section, we establish a surplus model for the insurer and formulate an optimization problem minimizing the ruin probability.
Section \ref{sec:grad_cal} is devoted to properties of the gradient of the objective function.
In Section \ref{sec:sgd}, we show a convergence result for the stochastic projected gradient method under the assumption of H\"older continuity of the gradient.
In Section \ref{sec:experiment}, we present numerical experiments demonstrating the effectiveness of our proposed method.
Finally, Section \ref{sec:conclusion} concludes the paper with some remarks.

\subsection{Notation and Preliminaries}\label{sec:notation}
The set of positive integers and the set of positive real numbers are denoted by $\mathbb{N}$ and $\mathbb{R}_{++}$, respectively.
For an integer $n$, the set $[n]$ is defined by $[n] \coloneqq \{1, \cdots, n\}$.
The standard inner product of $x,y \in \mathbb{R}^d$ is denoted by $\innerprod<x,y>$.
The $\ell_2$ norm of $x \in \mathbb{R}^d$ is defined by $\|x\| \coloneqq \sqrt{\innerprod<x,x>}$ and the uniform norm of a function $f: \mathbb{R} \to \mathbb{R}$ is defined by $\|f\|_{\infty} \coloneqq \sup_{x \in \mathbb{R}} |f(x)|$.
The space of continuously differential functions $w$ on $[0,T]$ such that $w(0)=w(T)=0$ is denoted by $C_0^1([0,T])$.
Let $\dist(x, \mathcal{X})$ denote the distance between $x \in \mathbb{R}^d$ and $\mathcal{X} \subset \mathbb{R}^d$, that is, $\dist(x, \mathcal{X}) \coloneqq \inf \left\{ \|x-y\| ~|~ y \in \mathcal{X}\right\}$.
The partial derivative of $f:\mathbb{R}^d \to \mathbb{R}$ with respect to $\xi \in \mathbb{R}$ is denoted by $\partial_\xi f$.
For a convex function $f:\mathbb{R}^d \to (-\infty,\infty]$, the subdifferential of $f$ at $x\in \mathbb{R}^d$ is defined by
\begin{align}
    \partial f(x) \coloneqq \left\{g \in \mathbb{R}^d ~|~ \forall y \in \mathbb{R}^d,~f(y) \geq f(x) + \innerprod< g,y-x> \right\}.
\end{align}
The Euclidean projection of a point $x \in \mathbb{R}^d$ onto a nonempty closed convex set $\mathcal{X} \subset \mathbb{R}^d$ is defined by
\begin{align}\label{eq:projection_opt}
     \proj_{\mathcal{X}}(x) \coloneqq \argmin_{y \in \mathcal{X}}  \|x-y\|^2.
\end{align}
For any $\gamma > 0,~ g \in \mathbb{R}^d,~x \in \mathbb{R}^d$, we define the projected gradient mapping
\begin{align}\label{eq:gpg}
     P_{\mathcal{X}}(x,g,\gamma) \coloneqq \frac{1}{\gamma}(x-x^{+}),
\end{align}
where $x^{+} \coloneqq \proj_{\mathcal{X}}(x-\gamma g)$.

The following inequality is frequently used in Section \ref{sec:grad_cal}.
\begin{lemma}\label{lemma_0}
For any $\{x_i\}_{i\in[n]} \subset \mathbb{R}$ and $p>0$, it holds that
    \begin{align}\label{:eq:lm0}
        \left|\sum_{i=1}^{n} x_i\right|^p \le n^p \sum_{i=1}^{n} \left|x_i\right|^p
    \end{align}
\end{lemma}

\begin{proof}
We have
\begin{align}
    \left|\sum_{i=1}^{n} x_i\right|^p 
    \le \left(\sum_{i=1}^{n} |x_i| \right)^p
    \le \left(\sum_{i=1}^{n} \max\{|x_1|,\ldots,|x_n|\}\right)^p
    &\le n^p \max\{|x_1|,\ldots,|x_n|\}^p \\
    &= n^p \max\{|x_1|^p,\ldots,|x_n|^p\} \\
    &\le n^p \sum_{i=1}^{n} \left|x_i\right|^p,
\end{align}
where the first inequality follows from the triangle inequality and the second one from the monotonicity of $(\cdot)^p$.
\end{proof}
The following lemmas are used in the convergence analysis in Section \ref{sec:sgd}.
They are the special cases of Lemma 1 and Proposition 1 in \citep{ghadimi2016mini}.

\begin{lemma}\label{lemma_1}
For any $x\in \mathcal{X},~g\in \mathbb{R}^d$ and $\gamma>0$, we have
    \begin{align}\label{:eq:lm1}
        \innerprod< g,P_\mathcal{X}(x,g,\gamma) > \geq  \|P_\mathcal{X}(x,g,\gamma) \|^2.
    \end{align}
\end{lemma}

\begin{lemma}\label{lemma_2}
For any $g_1,~ g_2 \in \mathbb{R}^d$, we have
   \begin{align}\label{eq:prop1}
       \| P_{\mathcal{X}}(x,g_1,\gamma) - P_{\mathcal{X}}(x,g_2,\gamma) \| \le  \| g_1 - g_2  \|.
   \end{align}
\end{lemma}

\section{Formulation of ruin probability minimization problem}\label{sec:model}
In this section, we define a surplus process of the insurer and then consider an optimization problem minimizing the ruin probability \eqref{eq:ruin_prob}.
The underlying probability space is denoted by $(\Omega,\mathcal{F},\mathbb{P})$.
Let $\{N_t\}_{t \geq 0}$ be the Poisson process with intensity $\lambda > 0$, the claim size $\{X_i\}_{i \in \mathbb{N}}$ be i.i.d. positive random variables independent of $\{N_t\}_{t \geq 0}$.
Suppose that the volume of claims varies with a constant inflation or deflation rate $r \neq 0$. Then, the surplus process $\{ U_t\}_{t \geq 0}$ is defined by
\begin{align}\label{eq:classic_model}
    \begin{split}
    U_t =~ & u + c_1t - \sum_{i=1}^{N_t}e^{rT_i}X_i,
    \end{split}
\end{align}
where $u > 0$ is the initial surplus, $c_1 > 0$ is the premium rate, and $\{T_i\}_{i\in \mathbb{N}}$ is jump times.
The process \eqref{eq:classic_model} is a classical risk model with inflation.
In \citep{taylor1979probability}, although the volume of premium income also varies with inflation, only the claim size varies with inflation in this paper because it would be unacceptable for the policyholder to vary premium income dynamically in practice. 
Please refer to \citep{grandell2012aspects} for the ruin theory in a classical risk model.

We consider that the insurer handles proportional reinsurance with retention level $b \in (0,1]$ and invests in $m \in \mathbb{N}$ assets.
Let $p = (p^{(1)},\ldots,p^{(m)})$ be investment ratio of initial surplus $u$ to $m$ assets satisfying $p^{(1)} + \cdots + p^{(m)}  = 1$ and $p^{(j)} \geq 0$ for all $j=1,\ldots,m$.
Then, the surplus process $U_t(p,b)$ is defined as
\begin{align}\label{eq:surplus_model}
    \begin{split}
    U_t(p,b) =~ & up^\top S_t+c_1t-(1-b)c_2t - 
    \sum_{i=1}^{N_t}be^{rT_i}X_i,
    \end{split}
\end{align}
where $c_2 > c_1 > 0$, $S_t = (S_t^{(1)},\ldots,S_t^{(n)})^\top$, $(1-b)c_2$ is the reinsurance premium rate, and $S_t^{(j)}>0$ is the price of $j$th asset at the time $t$ satisfying $S_0^{(j)}=1$, for example, the geometric Brownian motion.

One of the purposes of reinsurance is to reduce the insurer's ruin probability.
From this perspective, it is reasonable to determine the investment ratio $p$ and retention level $b$ that minimizes the ruin probability.
Thus, we consider solving the following optimization problem.

\begin{align}\label{eq:opt}
\begin{split}
\underset{(p,b)\in{\mathbb{R}^{m+1}}} {\text{minimize}} & \quad F(p,b)\coloneqq \mathbb{P}( U_T(p,b) < 0) \\
\text{subject to}  & \quad \underline{b}\le b\le 1,\\
& \quad  p^{(1)} + \cdots + p^{(m)}  = 1,\\
&  \quad   p^{(j)} \geq 0, ~~ j=1,\ldots,m,
\end{split}
\end{align}
where we assume $\underline{b} \ge 1-c_1/c_2 > 0$.
This is a realistic assumption such that the insurer never ruins without claims.
Therefore, the objective function $F(p,b)$ can be expressed as follows:
\begin{equation}\label{eq:obj_function}
F(p,b) = \mathbb{E}[\bm{1}_A\bm{1}_{\{U_T(p,b) < 0\}}],    
\end{equation}
where $A\coloneqq\{N_T\geq1\}$.
In what follows, $\mathcal{C}$ denotes the feasible region of the optimization problem \eqref{eq:opt}.

\section{Unbiased estimator of gradient of ruin probability}\label{sec:grad_cal}
If the gradient of \eqref{eq:obj_function} is computed, one can apply a gradient-type algorithm for solving the problem \eqref{eq:opt}.
However, the indicator function $\bm{1}_{\{U_T(p,b) < 0\}}$ in \eqref{eq:obj_function} is not only nondifferentiable but also discontinuous.
Consequently, it is difficult to interchange of differentiation and expectation.
Although an alternative method to estimate the gradient is the use of the finite difference, \citet{privault2004malliavin} showed that this approximation performs poorly when combined with the Monte Carlo method.
Referring to the idea of \citet{privault2004malliavin}, we consider obtaining the gradient by applying the Malliavin calculus for the Poisson process.
As a result, we have access to a computable unbiased estimator of the gradient and can exploit a stochastic gradient-based algorithm.
To this end, we make assumptions on $w:[0,T]\to\mathbb{R}$ that is used to construct the unbiased estimator.

\begin{assumption}\label{assumption_grad}
    ~
    
    \begin{enumerate}[(a)]
        \item $w \in C_0^1([0,T]),~ w(t)>0~ (0<t<T),~ w^{-1} \in L^4(0,T)$.
        \item $\mathbb{E}[X_1^{-4}] < \infty,\ \mathbb{E}\left[\left|S_T^{(j)}\right|^4\right] < \infty,~j=1,\ldots,m$.
    \end{enumerate}
\end{assumption}

The following lemma plays an important role in obtaining the unbiased estimator of the gradient.

\begin{lemma}\label{lemma_3}
Suppose Assumption \ref{assumption_grad} (a).
Let $\{Y_i\}_{i \in \mathbb{N}}$ be i.i.d. positive random variables independent of $\{N_t\}_{t \geq 0}$, and assume that $\mathbb{E}[Y_1^{-4}] < \infty$.
Then, it holds that 
\begin{align}
    \mathbb{E}\left[ \bm{1}_A\left( \sum_{i=1}^{N_T} w(T_i)^{-4}Y_i^{-4} \right)\right] = \lambda\mathbb{E}\left[ Y_1^{-4}\right] \int_{0}^{T} w(t)^{-4} dt.
\end{align}
\end{lemma}

\begin{proof}
From the independence of $\{Y_i\}_{i\in \mathbb{N}}$ and $\{N_t\}_{t\ge 0}$, we obtain
\begin{align}
\begin{split}
    \mathbb{E}\left[ \bm{1}_A\left( \sum_{i=1}^{N_T} w(T_i)^{-4}Y_i^{-4} \right)\right]
    &= \mathbb{E}\left[\sum_{n=1}^{\infty}  \left( \sum_{i=1}^{n} w(T_i)^{-4}Y_i^{-4} \right)\bm{1}_{\left\{N_T = n\right\}}
    \right]\\
    &=  \sum_{n=1}^{\infty}\mathbb{E}\left[  \sum_{i=1}^{n} w(T_i)^{-4}Y_i^{-4} \middle| N_T = n\right]\mathbb{P}(N_T = n)\\
    & = \mathbb{E}\left[ Y_1^{-4}\right]\sum_{n=1}^{\infty}\mathbb{E}\left[  \sum_{i=1}^{n} w(T_i)^{-4}\middle| N_T = n\right]\mathbb{P}(N_T = n),
\end{split}\label{eq:leq_lemma3_1}
\end{align}
where the second equality follows from the monotone convergence theorem.
The distribution of jump times $(T_1,\ldots,T_n)$ conditioned on $\{N_T = n\}$ coincides with the distribution of the order statistics $(U_{(1)},\ldots,U_{(n)})$ of i.i.d. random variables $(U_1,\ldots,U_n)$ from the uniform distribution on $[0,T]$ \citep[Exercise 2.1.5]{daley2003introduction}.
Thus, it follows that
\begin{align}
    \mathbb{E}\left[  \sum_{i=1}^{n} w(T_i)^{-4}\middle| N_T = n\right]
     = \mathbb{E}\left[  \sum_{i=1}^{n} w(U_{(i)})^{-4}\right]
    &= \mathbb{E}\left[  \sum_{i=1}^{n} w(U_{i})^{-4}\right]\\
    &= n \mathbb{E}\left[w(U_1)^{-4}\right]\\
    &= \frac{n}{T} \int_{0}^{T} w(t)^{-4} dt.
\end{align}
Consequently, plugging the above equality into the equality \eqref{eq:leq_lemma3_1}, we have
\begin{align}
   \mathbb{E}\left[ \bm{1}_A\left( \sum_{i=1}^{N_T} w(T_i)^{-4}Y_i^{-4} \right)\right]
    = \mathbb{E}\left[ Y_1^{-4}\right]\mathbb{E}\left[N_T\right]\frac{1}{T}\int_{0}^{T} w(t)^{-4} dt 
    = \lambda\mathbb{E}\left[ Y_1^{-4}\right]\int_{0}^{T} w(t)^{-4} dt.
\end{align}
\end{proof}

The following theorem provides the unbiased estimator of the gradient of \eqref{eq:obj_function}.

\begin{theorem}\label{theorem_grad}
    Suppose Assumption \ref{assumption_grad} and let $(p,b) \in \mathcal{C}$. Then, we have
    \begin{align}\label{eq:grad}
        \nabla F(p,b) = \mathbb{E}\left[\begin{pmatrix}
        \bm{1}_A W_{p^{(1)}}  \bm{1}_{\{U_T(p,b) < 0\}}\\
        \vdots\\
        \bm{1}_A W_{p^{(m)}}  \bm{1}_{\{U_T(p,b) < 0\}}\\
        \bm{1}_A W_{b}  \bm{1}_{\{U_T(p,b) < 0\}}\\
        \end{pmatrix}\right],
    \end{align}
where the weights $W_{p^{(j)}}$ and $W_b$ are given by
\begin{align}
    &W_{p^{(j)}} = \frac{uS_T^{(j)}}{br\sum_{i=1}^{N_T} w(T_i)e^{rT_i}X_i}\left( \sum_{i=1}^{N_T} w^{\prime}(T_i) - \frac{\sum_{i=1}^{N_T} w(T_i)\{rw(T_i) + w^{\prime}(T_i)\}e^{rT_i} X_i}{\sum_{i=1}^{N_T} w(T_i)e^{rT_i}X_i}\right),\\
    & W_{b} = \frac{c_2T - \sum_{i=1}^{N_T} e^{rT_i}X_i}{ br\sum_{i=1}^{N_T} w(T_i)e^{rT_i}X_i}\left( \sum_{i=1}^{N_T} w^{\prime}(T_i) - \frac{\sum_{i=1}^{N_T} w(T_i)\{rw(T_i) + w^{\prime}(T_i)\}e^{rT_i} X_i}{\sum_{i=1}^{N_T} w(T_i)e^{rT_i}X_i}\right) - \frac{1}{b}.
\end{align}
\end{theorem}

\begin{proof}
To apply Proposition 3.1 in \citep{privault2004malliavin},
we will indicate the fulfillment of the following conditions:
\begin{align}
\bm{1}_A\frac{\partial_b U_T(p,b)}{D_w U_T(p,b)},~\bm{1}_A\frac{\partial_{p^{(j)}} U_T(p,b)}{D_w U_T(p,b)}\in L^4(\Omega),~j=1,\ldots,m,
\end{align}
where
\begin{equation}
    \partial_{p^{(j)}} U_T (p,b) = uS_T^{(j)},~ \partial_b U_T (p,b) = c_2T - \sum_{i=1}^{N_T} e^{rT_i}X_i,~ D_w U_T (p,b) = b r \sum_{i=1}^{N_T} w(T_i)e^{rT_i}X_i,
\end{equation}
and $D_w$ is the gradient operator in the Malliavin calculus for the Poisson process.
See \citep{privault2004malliavin} for details.
We first show that the condition for $b$ is satisfied.
Lemma \ref{lemma_0} shows that
\begin{align}
    \mathbb{E}\left[ \left|\bm{1}_A\frac{\partial_b U_T(p,b)}{D_w U_T(p,b)} \right|^4\right] &=  \mathbb{E}\left[ \bm{1}_A \left| \frac{c_2 T - \sum_{i=1}^{N_T} e^{rT_i}X_i}{br\sum_{i=1}^{N_T} w(T_i)e^{rT_i}X_i} \right|^4\right] \\
    & \le \mathbb{E}\left[ \frac{16\bm{1}_A c_2^4T^4}{b^4r^4\left|\sum_{i=1}^{N_T} w(T_i)e^{rT_i}X_i \right|^4}\right] +
    \mathbb{E}\left[ \frac{16 \bm{1}_A}{b^4r^4}\left( \frac{\sum_{i=1}^{N_T} e^{rT_i}X_i}{\sum_{i=1}^{N_T} w(T_i) e^{rT_i}X_i }\right)^4\right].
\end{align}
From Assumption \ref{assumption_grad} and Lemma \ref{lemma_3}, the first term of the above inequality can be evaluated as
\begin{align}
    \mathbb{E}\left[ \frac{16\bm{1}_A c_2^4T^4}{b^4r^4\left|\sum_{i=1}^{N_T} w(T_i)e^{rT_i}X_i \right|^4}\right]
    &= \frac{16c_2^4T^4}{b^4r^4 }\mathbb{E}\left[ \bm{1}_A\left( \sum_{i=1}^{N_T} w(T_i)e^{rT_i}X_i \right)^{-4}\right]\\
    &\le \frac{16c_2^4T^4}{b^4r^4 \underline{\rho}^4}\mathbb{E}\left[ \bm{1}_A\left( \sum_{i=1}^{N_T} w(T_i)X_i \right)^{-4}\right]\\
    &\le \frac{16c_2^4T^4}{b^4r^4 \underline{\rho}^4}\mathbb{E}\left[ \bm{1}_A  w(T_1)^{-4}X_1^{-4}\right]\\
    &\le \frac{16c_2^4T^4}{b^4r^4 \underline{\rho}^4}\mathbb{E}\left[ \bm{1}_A\left( \sum_{i=1}^{N_T} w(T_i)^{-4}X_i^{-4} \right)\right]\\
    &= \frac{16 \lambda c_2^4T^4}{b^4r^4 \underline{\rho}^4} \mathbb{E}\left[ X_1^{-4}\right] \int_{0}^{T} w(t)^{-4} dt < \infty,
\end{align}
where $\underline{\rho}\coloneqq \min\left\{1,e^{rT}\right\}$.
The second term is evaluated as
\begin{align}
    \mathbb{E}\left[ \frac{16 \bm{1}_A}{b^4r^4}\left( \frac{\sum_{i=1}^{N_T} e^{rT_i}X_i}{\sum_{i=1}^{N_T} w(T_i) e^{rT_i}X_i }\right)^4\right] &= \mathbb{E}\left[ \frac{16 \bm{1}_A}{b^4r^4}\left\{ \sum_{i=1}^{N_T} \left(\frac{ e^{rT_i}X_i}{\sum_{i=1}^{N_T} e^{rT_i}X_i } \right)w(T_i)\right\}^{-4}\right]\\
    &\le \mathbb{E}\left[ \frac{16 \bm{1}_A}{b^4r^4} \sum_{i=1}^{N_T}\left\{ \left(\frac{ e^{rT_i}X_i}{\sum_{i=1}^{N_T} e^{rT_i}X_i } \right)\frac{1
    }{w(T_i)^4} \right\}\right]\\
    & \le  \frac{16}{b^4r^4} \mathbb{E}\left[\bm{1}_A\left(\sum_{i=1}^{N_T} w(T_i)^{-4} \right)\right]\\
    & = \frac{16 \lambda  }{b^4r^4} \int_{0}^{T} w(t)^{-4}dt < \infty,
\end{align}
where the first inequality follows from the Jensen's inequality for the convex function $(\cdot)^{-4}$, and the last equality follows from Lemma \ref{lemma_3} with $Y_i \equiv 1$.
Therefore, the condition for $b$ is fulfilled.
To show the fulfillment of the condition for $p^{(j)}$, we have
\begin{align}
    \mathbb{E}\left[ \left|\bm{1}_A\frac{\partial_{p^{(j)}} U_T}{D_w U_T}\right|^4\right] &= \mathbb{E}\left[ \left|\bm{1}_A\frac{uS_T^{(j)}}{br\sum_{i=1}^{N_T} w(T_i)e^{rT_i}X_i}\right|^4\right] \\
    &\le \frac{\mathbb{E}\left[\left|uS_T^{(j)}\right|^4\right]}{b^4r^4\underline{\rho}^4}\mathbb{E}\left[\bm{1}_A \left(\sum_{i=1}^{N_T} w(T_i)X_i\right)^{-4}\right]\\
    &\le \frac{\mathbb{E}\left[\left|uS_T^{(j)}\right|^4\right]}{b^4r^4\underline{\rho}^4}\mathbb{E}\left[ \bm{1}_A\left( \sum_{i=1}^{N_T} w(T_i)^{-4}X_i^{-4} \right)\right]\\
    & \le \frac{\lambda u^4\mathbb{E}\left[\left|S_T^{(j)}\right|^4\right]}{b^4r^4\underline{\rho}^4}\mathbb{E}\left[ X_1^{-4}\right] \int_{0}^{T} w(t)^{-4} dt < \infty,
\end{align}
where the last inequality follows from Lemma \ref{lemma_3}.
Consequently, the desired result follows from Proposition 3.1 in \citep{privault2004malliavin}. 
\end{proof}

Since the random vector in \eqref{eq:grad} consists of simple operations on the Poisson process, claim sizes, and other deterministic parameters, it can be generated.
The random vector is also unbiased, and projection onto the feasible region $\mathcal{C}$ can be obtained explicitly.
Therefore, we apply the stochastic projected gradient method to the ruin probability minimization problem \eqref{eq:opt}.
The following proposition establishes the boundedness of the second moment for the stochastic gradient in \eqref{eq:grad}. This is one of the assumptions required in the convergence analysis in Section \ref{sec:sgd}.

\begin{proposition}
    Suppose that Assumption \ref{assumption_grad} holds.
    Then, there exists $M_1 > 0$ such that
    \begin{align}
        \mathbb{E}\left[\left\|\begin{pmatrix}
        \bm{1}_A W_{p^{(1)}}  \bm{1}_{\{U_T(p,b) < 0\}}\\
        \vdots\\
        \bm{1}_A W_{p^{(m)}}  \bm{1}_{\{U_T(p,b) < 0\}}\\
        \bm{1}_A W_{b}  \bm{1}_{\{U_T(p,b) < 0\}}\\
        \end{pmatrix}\right\|^2\right] \le M_1
    \end{align}
    for all $(p,b) \in \mathcal{C}$.
\end{proposition}
\begin{proof}
By setting
\begin{align}
    &Q^{(j)} \coloneqq \frac{uS_T^{(j)}}{r\sum_{i=1}^{N_T} w(T_i)e^{rT_i}X_i}\left( \sum_{i=1}^{N_T} w^{\prime}(T_i) - \frac{\sum_{i=1}^{N_T} w(T_i)\{rw(T_i) + w^{\prime}(T_i)\}e^{rT_i} X_i}{\sum_{i=1}^{N_T} w(T_i)e^{rT_i}X_i}\right),~j=1,\ldots,m, \label{eq:Q}\\
    &R \coloneqq \frac{c_2T - \sum_{i=1}^{N_T} e^{rT_i}X_i}{ r\sum_{i=1}^{N_T} w(T_i)e^{rT_i}X_i}\left( \sum_{i=1}^{N_T} w^{\prime}(T_i) - \frac{\sum_{i=1}^{N_T} w(T_i)\{rw(T_i) + w^{\prime}(T_i)\}e^{rT_i} X_i}{\sum_{i=1}^{N_T} w(T_i)e^{rT_i}X_i}\right) - 1,\label{eq:R}
\end{align}
$W_{p^{(j)}}$ and $W_b$ can be expressed as $W_{p^{(j)}}=\frac{Q^{(j)}}{b}$ and $W_b = \frac{R}{b}$, respectively.
Accordingly, the square integrability of $\bm{1}_AW_{p^{(1)}},\ldots,\bm{1}_AW_{p^{(m)}},\bm{1}_AW_b$ \footnote{This comes from the definition of the adjoint operator on the Hilbert space. See \citep[proof of Proposition 3.1]{privault2004malliavin} for details.} implies that of $\bm{1}_AQ^{(1)},\ldots,\bm{1}_AQ^{(m)},\bm{1}_AR$.
Thus, we have
\begin{align}
    \mathbb{E}\left[\left\|\begin{pmatrix}
    \bm{1}_A W_{p^{(1)}}  \bm{1}_{\{U_T(p,b) < 0\}}\\
    \vdots\\
    \bm{1}_A W_{p^{(m)}}  \bm{1}_{\{U_T(p,b) < 0\}}\\
    \bm{1}_A W_{b}  \bm{1}_{\{U_T(p,b) < 0\}}\\
    \end{pmatrix}\right\|^2\right] &\le \mathbb{E}\left[ \bm{1}_A W_{p^{(1)}}^2\right]+\cdots+\mathbb{E}\left[ \bm{1}_A W_{p^{(m)}}^2\right]+\mathbb{E}\left[ \bm{1}_A W_b^2\right] \\
    &\le \underbrace{\frac{1}{\underline{b}}\left(\mathbb{E}\left[ \bm{1}_A {Q^{(1)}}^2\right]+\cdots+\mathbb{E}\left[ \bm{1}_A {Q^{(m)}}^2\right]+\mathbb{E}\left[ \bm{1}_A R^2\right]\right)}_{\eqqcolon M_1}
    < \infty.
\end{align}
This completes the proof.
\end{proof}

Convergence analysis for gradient-based methods often relies on Lipschitz continuity and more generally on H\"older continuity.
To establish the H\"older continuity of the gradient \eqref{eq:grad} on the feasible region $\mathcal{C}$, we make the following assumption.

\begin{assumption}\label{integral_assump}
There exists $q>0$ such that $\mathbb{E}\left[ \left|S_T^{(j)}\right|^q \right] < \infty,~j=1,\ldots,m,~\mathbb{E}\left[ \left|X_1\right|^q \right] < \infty$.
\end{assumption}

We first show the following Avikainen's estimate-type inequality, which is a slight extension of \citep{avikainen2009irregular} for the conditional density function.

\begin{lemma}\label{avikainen}
Let $(p_1,b_1),(p_2,b_2)\in \mathcal{C}$ and assume that the conditional density $f_{U_T\left(p_1,b_1\right)|A}$ exists.
For any $\alpha,~\beta > 0$, we have
\begin{align}
    &\mathbb{E}\left[ \left| \bm{1}_{\left\{U_T\left(p_1,b_1\right) < 0\right\}} - \bm{1}_{\left\{U_T\left(p_2,b_2\right) < 0\right\}}\right|^{\alpha}\right]\\
    &\le 2^{\frac{2\beta+1}{1+\beta}}\left(1-e^{-\lambda T}\right)^{\frac{\beta}{1+\beta}}\| f_{U_T\left(p_1,b_1\right)|A} \|_{\infty}^{\frac{\beta}{1+\beta}} \mathbb{E}\left[ \left| U_T\left(p_1,b_1\right) - U_T\left(p_2,b_2\right) \right|^\beta \right]^{\frac{1}{1+\beta}}.
\end{align}
\end{lemma}

\begin{proof}
Let $A_1 \coloneqq \left\{U_T\left(p_1,b_1\right) < 0, U_T\left(p_2,b_2\right) \geq 0\right\},A_2 \coloneqq \left\{U_T\left(p_1,b_1\right) \geq 0, U_T\left(p_2,b_2\right) < 0\right\}$. Then, it follows that
\begin{align}\label{eq:lem_avi_1}
\begin{split}
    \mathbb{E}\left[ \left| \bm{1}_{\left\{U_T\left(p_1,b_1\right) < 0\right\}} - \bm{1}_{\left\{U_T\left(p_2,b_2\right) < 0\right\}}\right|^{\alpha}\right] &=\mathbb{E}\left[ \left| \bm{1}_{\left\{U_T\left(p_1,b_1\right) < 0\right\}} - \bm{1}_{\left\{U_T\left(p_2,b_2\right) < 0\right\}}\right|\right]\\
    &= \mathbb{P} (A_1 \cup A_2)\\
    &= \mathbb{P} ((A_1 \cup A_2) \cap A).
\end{split}
\end{align}
For any $\varepsilon > 0$, it holds that $A_1 \cup A_2 \subset \left\{\left|U_T\left(p_1,b_1\right) \right|\le \varepsilon\right\} \cup\left\{\left|U_T\left(p_1,b_1\right)-U_T\left(p_2,b_2\right)\right| \geq \varepsilon \right\}$. Thus, the equality \eqref{eq:lem_avi_1} can be evaluated as
\begin{align}\label{eq:lem_avi_2}
\begin{split}
    &\mathbb{E}\left[ | \bm{1}_{\left\{U_T\left(p_1,b_1\right) < 0\right\}} - \bm{1}_{\left\{U_T\left(p_2,b_2\right) < 0\right\}}|^{\alpha}\right]\\
    &\le \mathbb{P}\left( \left|U_T\left(p_1,b_1\right)\right| \le  \varepsilon \cap A\right)+\mathbb{P}\left(\left|U_T\left(p_1,b_1\right)-U_T\left(p_2,b_2\right)\right| \geq \varepsilon \cap A\right)\\
    &= \mathbb{P}\left( \left|U_T\left(p_1,b_1\right)\right| \le  \varepsilon \middle| A\right)\mathbb{P}(A)+\mathbb{P}\left(\left|U_T\left(p_1,b_1\right)-U_T\left(p_2,b_2\right)\right| \geq \varepsilon \cap A\right)\\
    &\le \mathbb{P}(A)\int_{-\varepsilon}^{\varepsilon} f_{U_T\left(p_1,b_1\right) | A}(x) dx + \mathbb{P}\left(\left|U_T\left(p_1,b_1\right)-U_T\left(p_2,b_2\right)\right| \geq \varepsilon\right)\\
    &\le 2\varepsilon\mathbb{P}(A)\|f_{U_T\left(p_1,b_1\right) | A}\|_\infty + \mathbb{P}\left(\left|U_T\left(p_1,b_1\right)-U_T\left(p_2,b_2\right)\right| \geq \varepsilon\right)\\
    & \le 2\varepsilon\mathbb{P}(A)\|f_{U_T\left(p_1,b_1\right) | A}\|_\infty + \frac{1}{\varepsilon^\beta}\mathbb{E}\left[ \left| U_T\left(p_1,b_1\right) - U_T\left(p_2,b_2\right)\right|^\beta \right],
\end{split}
\end{align}
where the last inequality follows from the Markov's inequality.
Setting
\begin{equation}
    \varepsilon = \left( 2\left\|f_{U_T\left(p_1,b_1\right)|A}\right\|_{\infty} \mathbb{P}(A) \right)^{-\frac{1}{1+\beta}}\mathbb{E}\left[ \left| U_T\left(p_1,b_1\right) - U_T\left(p_2,b_2\right)\right|^\beta \right]^{\frac{1}{1+\beta}}
\end{equation}
yields
\begin{align}
    &\mathbb{E}\left[ \left| \bm{1}_{\left\{U_T\left(p_1,b_1\right) < 0\right\}} - \bm{1}_{\left\{U_T\left(p_2,b_2\right) < 0\right\}}\right|^{\alpha}\right]\\
    &\le 2^{\frac{2\beta+1}{1+\beta}}\mathbb{P}(A)^{\frac{\beta}{1+\beta}}\left\| f_{U_T\left(p_1,b_1\right)|A} \right\|_{\infty}^{\frac{\beta}{1+\beta}} \mathbb{E}\left[ \left| U_T\left(p_1,b_1\right) - U_T\left(p_2,b_2\right)\right|^\beta \right]^{\frac{1}{1+\beta}}\\
    &= 2^{\frac{2\beta+1}{1+\beta}}\left(1-e^{-\lambda T}\right)^{\frac{\beta}{1+\beta}}\left\| f_{U_T\left(p_1,b_1\right)|A} \right\|_{\infty}^{\frac{\beta}{1+\beta}} \mathbb{E}\left[ \left| U_T\left(p_1,b_1\right) - U_T\left(p_2,b_2\right)\right|^\beta \right]^{\frac{1}{1+\beta}},
\end{align}
which is the desired result.
\end{proof}

Next, the following lemma implies that the conditional density function of the surplus $U_T(p,b)$ is bounded on the feasible region $\mathcal{C}$.

\begin{lemma}\label{dense_bound}
    Suppose that Assumption \ref{assumption_grad} holds.
    Then, there exists $M_2>0$ such that
    \begin{equation}
        \|f_{U_T(p,b)|A}\|_\infty \le M_2
    \end{equation}
    for all $(p,b)\in\mathcal{C}$.
\end{lemma}

\begin{proof}
The density of $U_T(p,b)$ conditioned on $A$ can be expressed as
\begin{align}\label{eq:dens_func}
\begin{split}
    f_{U_T(p,b)|A}(y) = \frac{\partial}{\partial y}\mathbb{P}(U_T(p,b)\le y|A)
    &= \frac{1}{\left(1-e^{-\lambda T}\right)}\frac{\partial}{\partial y}\mathbb{P}(U_T(p,b)\le y\cap A)\\
    &= \frac{1}{\left(1-e^{-\lambda T}\right)}\frac{\partial}{\partial y}\mathbb{E}\left[ \bm{1}_A \bm{1}_{\left\{ U_T(p,b) \le y\right\}} \right]\\
    &= \frac{1}{\left(1-e^{-\lambda T}\right)}\frac{\partial}{\partial y}\mathbb{E}\left[ \bm{1}_A \bm{1}_{\left\{ V^y_T(p,b) \le 0\right\}} \right],
\end{split}
\end{align}
where $V^y_T(p,b) \coloneqq U_T(p,b)-y$.
We will indicate the fulfillment of the condition $\bm{1}_A\frac{\partial_y V^y_T(p,b)}{D_w V^y_T(p,b)} \in L^4(\Omega)$ to apply Proposition 3.1 in \citep{privault2004malliavin}.
It follows from $\partial_y V^y_T(p,b)=-1$ and $D_w V^y_T(p,b) = br \sum_{i=1}^{N_T} w(T_i)e^{rT_i}X_i$ that
\begin{align}
    \mathbb{E}\left[ \left|\bm{1}_A\frac{\partial_{y} V^y_T(p,b)}{D_wV^y_T(p,b)}\right|^4\right] &= \frac{1}{b^4 r^4}\mathbb{E}\left[\bm{1}_A \left(\sum_{i=1}^{N_T} w(T_i)e^{rT_i}X_i\right)^{-4}\right] \\
    &\le  \frac{1}{b^4 r^4 \underline{\rho}^4}\mathbb{E}\left[\bm{1}_A \left(\sum_{i=1}^{N_T} w(T_i)X_i\right)^{-4}\right] \\
    &\le \frac{1}{b^4r^4\underline{\rho}^4}\mathbb{E}\left[ \bm{1}_A\left( \sum_{i=1}^{N_T} w(T_i)^{-4}X_i^{-4} \right)\right]\\
    & \le \frac{\lambda}{b^4r^4\underline{\rho}^4}\mathbb{E}\left[ X_1^{-4}\right] \int_{0}^{T} w(t)^{-4} dt < \infty,
\end{align}
where the last inequality follows from Lemma \ref{lemma_3}.
Thus, we obtain from Proposition 3.1 in \citep{privault2004malliavin} that
\begin{align}
    \frac{\partial}{\partial y}\mathbb{E}\left[ \bm{1}_A \bm{1}_{\left\{ V^y_T(p,b) \le 0\right\}} \right]
    = \mathbb{E}\left[\bm{1}_{A}W_y\bm{1}_{\{V^y_T(p,b)\le 0\}} \right],
\end{align}
where $W_y$ is given by
\begin{align}
    W_y \coloneqq \frac{-1}{br \sum_{i=1}^{N_T} w(T_i)e^{rT_i}X_i}\left(\sum_{i=1}^{N_T} w^{\prime}(T_i) - \frac{\sum_{i=1}^{N_T} w(T_i)\{rw(T_i) + w^{\prime}(T_i)\}e^{rT_i} X_i}{\sum_{i=1}^{N_T} w(T_i)e^{rT_i}X_i} \right).
\end{align}
Combining the equality \eqref{eq:dens_func} with the above, we obtain from the square integrability of $\bm{1}_A W_y$ that
\begin{align}
     &\left|f_{U_T(p,b)|A}(y)\right|\\
     &= \frac{1}{\left(1-e^{-\lambda T}\right)} \left|\mathbb{E}\left[ \bm{1}_A W_y\bm{1}_{\left\{ V^y_T(p,b) \le 0\right\}} \right]\right|\\
     &\le \frac{1}{\left(1-e^{-\lambda T}\right)} \left|\mathbb{E}\left[ \frac{\bm{1}_A\bm{1}_{\left\{ V^y_T(p,b) \le 0\right\}}}{br \sum_{i=1}^{N_T} w(T_i)e^{rT_i}X_i}\left(\sum_{i=1}^{N_T} w^{\prime}(T_i) - \frac{\sum_{i=1}^{N_T} w(T_i)\{rw(T_i) + w^{\prime}(T_i)\}e^{rT_i} X_i}{ \sum_{i=1}^{N_T} w(T_i)e^{rT_i}X_i }\right) \right]\right|\\
     &\le \frac{1}{\underline{b}\left(1-e^{-\lambda T}\right)}\mathbb{E}\left[ \left|\frac{\bm{1}_A}{r\sum_{i=1}^{N_T} w(T_i)e^{rT_i}X_i}\left(\sum_{i=1}^{N_T} w^{\prime}(T_i) - \frac{\sum_{i=1}^{N_T} w(T_i)\{rw(T_i) + w^{\prime}(T_i)\}e^{rT_i} X_i}{ \sum_{i=1}^{N_T} w(T_i)e^{rT_i}X_i }\right)\right| \right]\\
     &\eqqcolon M_2 < \infty,
\end{align}
which implies the desired result.
\end{proof}

Finally, we prove the following inequality that is useful for providing H\"older continuity of the gradient.

\begin{lemma}\label{lemma_3_4}
Suppose that Assumption \ref{integral_assump} holds. Then, for any $(p_1,b_1),(p_2,b_2)\in \mathcal{C}$, we have
\begin{align}
    \mathbb{E}\left[ \left| U_T\left(p_1,b_1\right) - U_T\left(p_2,b_2\right)\right|^q\right]^{\frac{1}{1+q}} \le M_3 \|(p_1,b_1)-(p_2,b_2)\|^{\frac{q}{1+q}},
\end{align}
where $M_3$ is given by
\begin{align}
    M_3\coloneqq (m+1)\max\left\{\mathbb{E}\left[\left|uS_T^{(1)}\right|^q\right],\ldots,\mathbb{E}\left[\left|uS_T^{(m)}\right|^q\right],\mathbb{E}\left[\left| c_2T- \sum_{i=1}^{N_T} e^{rT_i}X_i\right|^q\right] \right\}^{\frac{1}{1+q}}.
\end{align}
\end{lemma}

\begin{proof}
Let $(p_1,b_1), (p_2,b_2)\in \mathcal{C}$.
From Lemma \ref{lemma_0}, it holds that
\begin{align}\label{eq:le_3_4}
\begin{split}
     & \mathbb{E}\left[ \left| U_T\left(p_1,b_1\right) - U_T\left(p_2,b_2\right)\right|^q\right]^{\frac{1}{1+q}} \\
     & = \mathbb{E}\left[ \left| \sum_{j=1}^m uS_T^{(j)}(p^{(j)}_1 - p^{(j)}_2)+\left( c_2T- \sum_{i=1}^{N_T} e^{rT_i}X_i\right)(b_2-b_1)\right|^q\right]^{\frac{1}{1+q}}\\
     &\le (m+1)^{\frac{q}{1+q}}\mathbb{E}\left[ \sum_{j=1}^m \left|uS_T^{(j)}\right|^q\left|p^{(j)}_1 - p^{(j)}_2\right|^q+\left| c_2T- \sum_{i=1}^{N_T} e^{rT_i}X_i\right|^q\left|b_1-b_2\right|^q\right]^{\frac{1}{1+q}}\\
     & \le  (m+1)^{\frac{q}{1+q}}\max\left\{\mathbb{E}\left[\left|uS_T^{(1)}\right|^q\right],\ldots,\mathbb{E}\left[\left|uS_T^{(m)}\right|^q\right],\mathbb{E}\left[\left| c_2T- \sum_{i=1}^{N_T} e^{rT_i}X_i\right|^q\right] \right\}^{\frac{1}{1+q}}\\
     &\quad \times \left( \left|p^{(1)}_1 - p^{(1)}_2\right|^q+\cdots +\left|p^{(m)}_1 - p^{(m)}_2\right|^q+\left|b_1-b_2\right|^q \right)^{\frac{1}{1+q}}.
\end{split}
\end{align}
Here, $\mathbb{E}\left[\left| c_2T- \sum_{i=1}^{N_T} e^{rT_i}X_i\right|^q\right]$ is integrable because we obtain from Lemma \ref{lemma_0} and Assumption \ref{integral_assump} that
\begin{align}
    \mathbb{E}\left[\left| c_2T- \sum_{i=1}^{N_T} e^{rT_i}X_i\right|^q\right] \le (2c_2 T )^q + 2^q \bar{\rho}^q \mathbb{E}\left[\left|\sum_{i=1}^{N_T} X_i\right|^q\right]&\le (2c_2 T )^q + 2^q \bar{\rho}^q \mathbb{E}\left[N_T^q\sum_{i=1}^{N_T} \left|X_i\right|^q\right]\\
    &\le (2c_2 T )^q + 2^q \bar{\rho}^q \sum_{n=1}^{\infty} n^q\mathbb{E}\left[\sum_{i=1}^{n} \left|X_i\right|^q\right]\mathbb{P}(N_T = n)\\
    &\le (2c_2 T )^q + 2^q \bar{\rho}^q\mathbb{E}\left[ \left|X_1\right|^q\right]\mathbb{E}\left[N_T^{q+1}\right]< \infty,
\end{align}
where $\bar{\rho} \coloneqq \max\left\{1,e^{rT}\right\}$.
Furthermore, we have
\begin{align}
    &\left( \left|p^{(1)}_1 - p^{(1)}_2\right|^q+\cdots +\left|p^{(m)}_1 - p^{(m)}_2\right|^q+\left|b_1-b_2\right|^q \right)^{\frac{1}{1+q}}\\
    &\le (m+1)^{\frac{1}{1+q}} \max\left\{ \left|p^{(1)}_1 - p^{(1)}_2\right|^q,\ldots,\left|p^{(m)}_1 - p^{(m)}_2\right|^q,\left|b_1-b_2\right|^q \right\}^{\frac{1}{1+q}}\\
    &=  (m+1)^{\frac{1}{1+q}} \max\left\{ \left|p^{(1)}_1 - p^{(1)}_2\right|^2,\ldots,\left|p^{(m)}_1 - p^{(m)}_2\right|^2,\left|b_1-b_2\right|^2 \right\}^{\frac{q}{2(1+q)}}\\
    &\le  (m+1)^{\frac{1}{1+q}} \left\{ \left|p^{(1)}_1 - p^{(1)}_2\right|^2+\cdots+\left|p^{(m)}_1 - p^{(m)}_2\right|^2+\left|b_1-b_2\right|^2 \right\}^{\frac{q}{2(1+q)}}\\
    &= (m+1)^{\frac{1}{1+q}} \| (p_1,b_1) -(p_2,b_2)\|^{\frac{q}{1+q}}.
\end{align}
Consequently, combining the inequality \eqref{eq:le_3_4} with the above yields the following inequality:
\begin{align}
    & \mathbb{E}\left[ \left| U_T\left(p_1,b_1\right) - U_T\left(p_2,b_2\right)\right|^q\right]^{\frac{1}{1+q}} \\
    & \le (m+1)\max\left\{\mathbb{E}\left[\left|uS_T^{(1)}\right|^q\right],\ldots,\mathbb{E}\left[\left|uS_T^{(m)}\right|^q\right],\mathbb{E}\left[\left| c_2T- \sum_{i=1}^{N_T} e^{rT_i}X_i\right|^q\right] \right\}^{\frac{1}{1+q}} \| (p_1,b_1) -(p_2,b_2)\|^{\frac{q}{1+q}},
\end{align}
which is the desired result.
\end{proof}

With Lemma \ref{avikainen}, \ref{dense_bound}, and \ref{lemma_3_4}, we now establish H\"older continuity of the gradient \eqref{eq:grad}.

\begin{theorem}\label{thm:holder}
Suppose that Assumption \ref{assumption_grad} and Assumption \ref{integral_assump} hold.
Then, $\nabla F$ is $\frac{q}{2(1+q)}-$H\"older continuous on the feasible region $\mathcal{C}$.
\end{theorem}

\begin{proof}
Let $(p_1,b_1), (p_2,b_2)\in \mathcal{C}$.
We set $Q^{(j)}$ and $R$ as in \eqref{eq:Q} and \eqref{eq:R}, respectively.
We aim to evaluate
\begin{align}\label{eq:2_norm}
\begin{split}
    \| \nabla F(p_1,b_1) - \nabla F(p_2,b_2)  \| &= \Bigg\{ \sum_{j=1}^m\left(\mathbb{E}\left[ \bm{1}_A W_{p^{(j)}_1} \bm{1}_{\left\{U_T\left(p_1,b_1\right) < 0\right\}}\right] - \mathbb{E}\left[ \bm{1}_A W_{p^{(j)}_2} \bm{1}_{\left\{U_T\left(p_2,b_2\right) < 0\right\}}\right] \right)^2 \\
    &\qquad + \left(\mathbb{E}\left[ \bm{1}_A W_{b_1} \bm{1}_{\left\{U_T\left(p_1,b_1\right) < 0\right\}}\right] - \mathbb{E}\left[ \bm{1}_A W_{b_2} \bm{1}_{\left\{U_T\left(p_2,b_2\right) < 0\right\}}\right] \right)^2
    \Bigg\}^{\frac{1}{2}}. 
\end{split}
\end{align}
To begin with, we evaluate $j$th term of the right hand side of \eqref{eq:2_norm} for $j=1,\ldots,m$.
Noting $W_{p^{(j)}_1} = \frac{Q^{(j)}}{b_1}$ and $W_{p^{(j)}_2} = \frac{Q^{(j)}}{b_2}$, we have
\begin{align}\label{eq:p_inequality_1}
\begin{split}
    &\left| \mathbb{E}\left[ \bm{1}_A W_{p^{(j)}_1} \bm{1}_{\left\{U_T\left(p_1,b_1\right) < 0\right\}}\right] - \mathbb{E}\left[ \bm{1}_A W_{p^{(j)}_2} \bm{1}_{\left\{U_T\left(p_2,b_2\right) < 0\right\}}\right] \right| \\
    &\le  \mathbb{E}\left[  \left|  \bm{1}_A W_{p^{(j)}_1} \bm{1}_{\left\{U_T\left(p_1,b_1\right) < 0\right\}} -  \bm{1}_A W_{p^{(j)}_2} \bm{1}_{\left\{U_T\left(p_2,b_2\right) < 0\right\}}  \right| \right]\\
    &=  \mathbb{E}\left[  \left|  \bm{1}_A \bm{1}_{\left\{U_T\left(p_1,b_1\right) < 0\right\}} \left( W_{p^{(j)}_1}  - W_{p^{(j)}_2}  \right) + \bm{1}_A \left(W_{p^{(j)}_2} \bm{1}_{\left\{U_T\left(p_1,b_1\right) < 0\right\}} -  W_{p^{(j)}_2} \bm{1}_{\left\{U_T\left(p_2,b_2\right) < 0\right\}} \right) \right| \right]\\
    &\le  \mathbb{E}\left[ \bm{1}_A \left|  W_{p^{(j)}_1}- W_{p^{(j)}_2} \right| \right] + \mathbb{E}\left[\bm{1}_A \left| W_{p^{(j)}_2} \right| \left| \bm{1}_{\left\{U_T\left(p_1,b_1\right) < 0\right\}} - \bm{1}_{\left\{U_T\left(p_2,b_2\right) < 0\right\}}  \right| \right]\\
    &= \mathbb{E}\left[ \bm{1}_A |Q^{(j)}| \left|\frac{1}{b_1}-\frac{1}{b_2}\right|\right] + \mathbb{E}\left[\frac{\bm{1}_A|Q^{(j)}|}{b_2}\left|\bm{1}_{\left\{U_T\left(p_1,b_1\right) < 0\right\}} - \bm{1}_{\left\{U_T\left(p_2,b_2\right) < 0\right\}}\right|\right]\\
    &\le  \frac{ |b_1 - b_2|\mathbb{E}\left[\bm{1}_A |Q^{(j)}| \right]}{\underline{b}^2}   +  \frac{\mathbb{E}\left[ \bm{1}_A |Q^{(j)}|^2  \right]^{\frac{1}{2}}}{\underline{b}} E \left[ \left |\bm{1}_{\left\{U_T\left(p_1,b_1\right) < 0\right\}} - \bm{1}_{\left\{U_T\left(p_2,b_2\right) < 0\right\}}\right|^2 \right]^{\frac{1}{2}},
\end{split}
\end{align}
where the second inequality follows from the triangle inequality
and the second term in the last inequality follows from the Cauchy-Schwartz inequality.
For the second term of the right hand side of \eqref{eq:p_inequality_1}, combining Lemma \ref{avikainen} and \ref{dense_bound} yields
\begin{align}
      &\frac{\mathbb{E}\left[ \bm{1}_A |Q^{(j)}|^2  \right]^{\frac{1}{2}}}{\underline{b}}E \left[ \left |\bm{1}_{\left\{U_T\left(p_1,b_1\right) < 0\right\}} - \bm{1}_{\left\{U_T\left(p_2,b_2\right) < 0\right\}}\right|^2 \right]^{\frac{1}{2}}\\
      &\le  \frac{\sqrt{2^{\frac{2q+1}{1+q}}\left(1-e^{-\lambda T}\right)^{\frac{q}{1+q}}\| M_2^{\frac{q}{1+q}} \mathbb{E}\left[ \bm{1}_A |Q^{(j)}|^2  \right]}}{\underline{b}}\mathbb{E}\left[ \left| U_T\left(p_1,b_1\right) - U_T\left(p_2,b_2\right)\right|^q\right]^{\frac{1}{2(1+q)}}\\
     & \le \frac{\sqrt{2^{\frac{2q+1}{1+q}}\left(1-e^{-\lambda T}\right)^{\frac{q}{1+q}}\| M_2^{\frac{q}{1+q}} M_3  \mathbb{E}\left[ \bm{1}_A |Q^{(j)}|^2  \right]}}{\underline{b}} \| (p_1,b_1) -(p_2,b_2)\|^{\frac{q}{2(1+q)}},
\end{align}
where the last inequality follows from Lemma \ref{lemma_3_4}.
From the inequality \eqref{eq:p_inequality_1} and the above, we have
\begin{align}
     &\left| \mathbb{E}\left[ \bm{1}_A W_{p^{(j)}_1} \bm{1}_{\left\{U_T\left(p_1,b_1\right) < 0\right\}}\right] - \mathbb{E}\left[ \bm{1}_A W_{p^{(j)}_2} \bm{1}_{\left\{U_T\left(p_2,b_2\right) < 0\right\}}\right] \right| \\
    &\le  \frac{ |b_1 - b_2|\mathbb{E}\left[ \bm{1}_A |Q^{(j)}|  \right]}{\underline{b}^2} + \frac{\sqrt{2^{\frac{2q+1}{1+q}}\left(1-e^{-\lambda T}\right)^{\frac{q}{1+q}}\| M_2^{\frac{q}{1+q}} M_3  \mathbb{E}\left[ \bm{1}_A |Q^{(j)}|^2  \right]}}{\underline{b}} \| (p_1,b_1) -(p_2,b_2)\|^{\frac{q}{2(1+q)}}\\
    &\le \underbrace{\left(\frac{  |1 - \underline{b}|^{\frac{2+q}{2(1+q)}}\mathbb{E}\left[ \bm{1}_A |Q^{(j)}|  \right]}{\underline{b}^2} + \frac{\sqrt{2^{\frac{2q+1}{1+q}}\left(1-e^{-\lambda T}\right)^{\frac{q}{1+q}}\| M_2^{\frac{q}{1+q}} M_3  \mathbb{E}\left[ \bm{1}_A |Q^{(j)}|^2  \right]}}{\underline{b}} \right)}_{\eqqcolon L_j} \| (p_1,b_1) -(p_2,b_2)\|^{\frac{q}{2(1+q)}}.
\end{align}
Evaluating the last term of the right hand side of \eqref{eq:2_norm} in the same way, we obtain from $W_{b_1} = \frac{R}{b_1}$ and $W_{b_2} = \frac{R}{b_2}$ that
\begin{align}
    &\left| \mathbb{E}\left[ \bm{1}_A W_{b_1} \bm{1}_{\left\{U_T\left(p_1,b_1\right) < 0\right\}}\right] - \mathbb{E}\left[ \bm{1}_A W_{b_2} \bm{1}_{\left\{U_T\left(p_2,b_2\right) < 0\right\}}\right] \right| \\
    &\le  \mathbb{E}\left[ \bm{1}_A \left|  W_{b_1}- W_{b_2} \right| \right] + \mathbb{E}\left[\bm{1}_A \left| W_{b_2} \right| \left| \bm{1}_{\left\{U_T\left(p_1,b_1\right) < 0\right\}} - \bm{1}_{\left\{U_T\left(p_2,b_2\right) < 0\right\}}  \right| \right]\\
    &\le \mathbb{E}\left[ \bm{1}_A |R| \left|\frac{1}{b_1}-\frac{1}{b_2}\right|\right] + \mathbb{E}\left[\frac{|R|}{b_2}\left|\bm{1}_{\left\{U_T\left(p_1,b_1\right) < 0\right\}} - \bm{1}_{\left\{U_T\left(p_2,b_2\right) < 0\right\}}\right|\right]\\
    &\le  \frac{  |b_1 - b_2|\mathbb{E}\left[\bm{1}_A |R| \right]}{\underline{b}^2}   +  \frac{\mathbb{E}\left[ \bm{1}_A |R|^2  \right]^{\frac{1}{2}}}{\underline{b}} E \left[ \left |\bm{1}_{\left\{U_T\left(p_1,b_1\right) < 0\right\}} - \bm{1}_{\left\{U_T\left(p_2,b_2\right) < 0\right\}}\right|^2 \right]^{\frac{1}{2}}\\
    &\le  \underbrace{\left(\frac{  |1 - \underline{b}|^{\frac{2+q}{2(1+q)}}\mathbb{E}\left[ \bm{1}_A |R|  \right]}{\underline{b}^2} + \frac{\sqrt{2^{\frac{2q+1}{1+q}}\left(1-e^{-\lambda T}\right)^{\frac{q}{1+q}}\| M_2^{\frac{q}{1+q}} M_3 \mathbb{E}\left[ \bm{1}_A |R|^2  \right]}}{\underline{b}} \right)}_{\eqqcolon L_{m+1}} \| (p_1,b_1) -(p_2,b_2)\|^{\frac{q}{2(1+q)}}.
\end{align}
Consequently, the equality \eqref{eq:2_norm} can be evaluated as
\begin{align}
    &\| \nabla F(p_1,b_1) - \nabla F(p_2,b_2)  \| 
     \le \left\{ \sum_{j=1}^{m+1} L_j^2 \| (p_1,b_1) -(p_2,b_2)\|^{\frac{q}{1+q}} \right\}^{\frac{1}{2}}
     = \sqrt{ \sum_{j=1}^{m+1} L_j^2 } ~ \| (p_1,b_1) -(p_2,b_2)\|^{\frac{q}{2(1+q)}}.
\end{align}
This completes the proof.
\end{proof}

We see from Theorem \ref{thm:holder} that if any moments of $S_T^{(j)}$ and $X_i$ exist, then $\nabla F$ is $\nu$-H\"older continuous for any $\nu < \frac{1}{2}$.

\section{Stochastic projected gradient method under H\"older condition}\label{sec:sgd}
In this section, we consider the following general constrained optimization problems:
\begin{align}
\underset{x\in\mathcal{X}} {\text{minimize}} & \quad f(x), \label{eq:obj_f}
\end{align}
where the function $f:\mathbb{R}^d \rightarrow \mathbb{R}$ is differentiable on $\mathcal{X}$, not necessarily convex, and $\mathcal{X}$ is closed convex set in Euclidean space $\mathbb{R}^d$.
Note that the problem \eqref{eq:obj_f} includes the ruin probability minimization problem \eqref{eq:opt}.
We assume that there exists $f^*$ such that $f^* \le f(x)$ for any $x \in \mathcal{X}$ and the exact gradient $\nabla f $ of the objective function $f$ is not available.
While Lipschitz continuity of the gradient is a standard assumption (e.g., \citep{ghadimi2013stochastic,ghadimi2016mini}), we relax this assumption to H\"older continuity.
\begin{assumption}\label{holder_con}
 Let $\nu \in (0,1]$. $\nabla f$ is $\nu$-H\"older continuous on $\mathcal{X}$:~there exists $L>0,~\nu \in (0,1]$ such that
\begin{align}       
    \|\nabla f(y) - \nabla f(x) \| \le L \| y -x\|^{\nu}
\end{align}
for any $x,~y\in \mathcal{X}$.
\end{assumption}

We consider the mini-batch stochastic projected gradient method (Algorithm \ref{SPG}) for the problem \eqref{eq:obj_f}. We assume that a stochastic gradient $G(x_k,\xi_{k,i})$ can be generated at iteration $k$ of the algorithm, given $x_k$, where $\{\xi_{k,i}\}_{i \in [m_k]}$ are i.i.d. random vector independent of $\{x_0,x_1,\ldots,x_k\}$.

\begin{algorithm}                     
    \caption{Mini-batch stochastic projected gradient method }         
    \label{SPG}                          
    \begin{algorithmic}      
    \STATE {\bfseries Input:} $x_0 \in \mathcal{X}$, stepsizes $\{\gamma_k\}\subset\mathbb{R}_{++}$, batch sizes $\{m_k\}\subset\mathbb{N}$ and $k=0$.
    \REPEAT
    \STATE Generate stochastic gradients $\{G(x_k,\xi_{k,i})\}_{i\in[m_k]}$.
    \STATE Calculate $G_k = \frac{1}{m_k}\sum_{i \in [m_k]} G(x_k,\xi_{k,i})$.
    \STATE Compute $x_{k+1} = \proj_{\mathcal{X}}(x_k - \gamma_k G_k)$.
    \STATE Set $k \leftarrow k+1$.
    \UNTIL Terminated criterion is satisfied.
    \end{algorithmic}
\end{algorithm}

\citet{lei2019stochastic}, \citet{patel2021stochastic}, and \citet{patel2022global} provide a theoretical analysis of the stochastic gradient method under the assumption of H\"older continuity.
\citet{lei2019stochastic} assume that the functions inside the expectation of the object function are H\"older continuous and establish a convergence result for the unconstrained problem.
Since the problem \eqref{eq:opt} we would like to solve is constrained, and the function inside the expectation is not only non-differentiable but also discontinuous, and hence it is necessary to impose a more relaxed H\"older continuity.
While \citet{patel2021stochastic} and \citet{patel2022global} consider Assumption \ref{holder_con}, they mainly investigate almost sure convergence for unconstrained problems and do not provide convergence rates. Therefore, we establish a convergence analysis under Assumption \ref{holder_con} and provide convergence rates.

The problem \eqref{eq:obj_f} can be equivalently rewritten to
\begin{align}
\underset{x\in\mathbb{R}^d} {\text{minimize}} & \quad h(x) \coloneqq f(x) + \delta_{\mathcal{X}}(x), \label{eq:obj_h}
\end{align}
where the function $\delta_{\mathcal{X}}(x)$ is given by
\begin{align}
    \delta_{\mathcal{X}}(x) \coloneqq \left\{
    \begin{alignedat}{2}
        & 0,    &\quad&   x \in \mathcal{X},   \\
        & +\infty, &\quad&  x \notin \mathcal{X}.
    \end{alignedat}
    \right.
\end{align}
Any local minimizer $x^*$ of \eqref{eq:obj_f} satisfies
\begin{equation}
    0 \in \partial h(x^*) \coloneqq \nabla f(x^*) + \partial \delta_{\mathcal{X}}(x^*),
\end{equation}
which is a standard necessary optimality condition (see, e.g., \citep[Theorem 3.72 (a)]{beck2017first}).
In view of this, the optimality measure $\dist(0,\partial h(x))$ is exploited in this paper.

To establish a convergence analysis, we make the following assumptions for the stochastic gradient $G(x_k,\xi_{k,i})$.
The first assumption requires the stochastic gradients to be unbiased.
The second assumption is slightly stronger than the standard one: $\mathbb{E}[\|G(x_k,\xi_{k,i})- \nabla f(x_k)\|^2| \mathcal{F}_k] \le \sigma^2$.

\begin{assumption}\label{assumption_alg}
    For any non-negative integer $k,i$, we have
    \begin{enumerate}[(a)]
        \item $\mathbb{E}[G(x_k,\xi_{k,i})| \mathcal{F}_k] = \nabla f(x_k)$\label{eq:A},
        \item $\mathbb{E}[\|G(x_k,\xi_{k,i})\|^2| \mathcal{F}_k] \le
        \sigma^2$,\label{eq:B}
    \end{enumerate}
    where $\sigma > 0$ and $\{\mathcal{F}_k\}$ is the natural filtration of $\{x_k\}$.
\end{assumption}

The following lemma is an immediate consequence of the fundamental theorem of calculus (e.g., \citep[Lemma 1]{yashtini2016global}).

\begin{lemma}\label{hol_leq}
    Suppose Assumption \ref{holder_con}.
    Then, for any $x,~y\in \mathcal{X}$, we have
    \begin{align}\label{eq:holleq}
        f(y) \le f(x) + \innerprod<\nabla f(x),y-x > + \frac{L}{1+\nu} \|y - x\|^{1+\nu}.
    \end{align}
\end{lemma}

We first provide an upper bound of the weighted sum of the projected gradient mapping $\tilde{g}_{k} \coloneqq P_{\mathcal{X}}(x_k,G_k,\gamma_k)$.
To prove the result, we require the following lemma.

\begin{lemma}\label{lemma_4}
Suppose Assumption \ref{assumption_alg}.
We obtain the following two inequalities.
\begin{enumerate}[(a)]
    \item $\mathbb{E}[ \|\nabla f(x_k)-G_k\|^2] \le \frac{\sigma^2}{m_k}$.
    \item $\mathbb{E}[\|G_k\|^{1+\nu}] \le \sigma^{1+\nu}$.
\end{enumerate}
\end{lemma}
\begin{proof}
First, we prove the inequality (a). Setting $\delta_{k,i} \coloneqq \nabla f(x_k)-G(x_k,\xi_{k,i})$, it holds that
\begin{align}
    \mathbb{E}[ \|\nabla f(x_k)-G_k\|^2|\mathcal{F}_k] &= \frac{1}{m_k^2}\mathbb{E}\left[ \left\| \sum_{i \in [m_k]} \left(\nabla f(x_k) - G(x_k,\xi_{k,i})\right)\right\|^2 \middle| \mathcal{F}_k\right] \\
    &=\frac{1}{m_k^2} \mathbb{E}\left[\left\| \sum_{i \in [m_k]} \delta_{k,i}\right\|^2 \middle| \mathcal{F}_k\right] \\
    &=\frac{1}{m_k^2}  \mathbb{E}\left[\left\innerprod< \sum_{i \in [m_k]} \delta_{k,i},\sum_{i \in [m_k]} \delta_{k,i} \right> \middle| \mathcal{F}_k\right] \\
    &=\frac{1}{m_k^2}  \mathbb{E}\left[
    \sum_{i \in [m_k]}\|\delta_{k,i}\|^2 + \sum_{i\neq j }
    \left\innerprod< \delta_{k,i},\delta_{k,j} \right> \middle| \mathcal{F}_k\right]\\
    & = \frac{1}{m_k} \mathbb{E}[\|\delta_{k,1}\|^2|\mathcal{F}_k] +\frac{1}{m_k^2}\sum_{i\neq j } \left\innerprod< \mathbb{E}[\delta_{k,i}|\mathcal{F}_k],\mathbb{E}[\delta_{k,j}|\mathcal{F}_k] \right>\\
    &= \frac{1}{m_k}    \mathbb{E}\left[
   \|\delta_{k,1}\|^2 \middle| \mathcal{F}_k\right] \le \frac{1}{m_k}  \mathbb{E}\left[
   \|G(x_k,\xi_{k,1})\|^2 \middle| \mathcal{F}_k\right] \le \frac{\sigma^2}{m_k},
\end{align} 
where the last equality holds since $\mathbb{E}[\delta_{k,i}|\mathcal{F}_k]=0$ from Assumption \ref{assumption_alg} (a), and the last inequality follows from Assumption \ref{assumption_alg} (b).
By taking the expectation of the above, we have
\begin{align}
   \mathbb{E}[ \|\nabla f(x_k)-G_k\|^2] \le \frac{\sigma^2}{m_k}.
\end{align}
Finally, we show the inequality (b).
By using H\"older's inequality, we obtain
\begin{align}
    \mathbb{E}[\|G_k\|^{1+\nu}] \le  \mathbb{E}[\|G_k\|^2]^\frac{1+\nu}{2} = \mathbb{E}\left[\mathbb{E}\left[\|G_k\|^2 \middle| \mathcal{F}_k\right]\right]^\frac{1+\nu}{2}.\label{eq:lem_4}
\end{align}
On the other hand, we have
\begin{align}\label{eq:lem_4_1}
\begin{split}
    \mathbb{E}\left[\|G_k\|^2 \middle| \mathcal{F}_k\right] 
  &=\mathbb{E}\left[\frac{1}{m_k^2}\left\|\sum_{i \in [m_k]}G(x_k,\xi_{k,i})\right\|^2 \middle| \mathcal{F}_k\right]\\
  &\le \mathbb{E}\left[\frac{1}{m_k^2}\left(\sum_{i \in [m_k]}\|G(x_k,\xi_{k,i})\|\right)^2 \middle| \mathcal{F}_k \right] \\
  &= \mathbb{E}\left[\frac{1}{m_k^2}\left(\sum_{i \in [m_k]}\|G(x_k,\xi_{k,i})\|^2 +\sum_{i\neq j }\|G(x_k,\xi_{k,i}) \|\|G(x_k,\xi_{k,j})\|\right) \middle| \mathcal{F}_k \right]\\
  &= \frac{1}{m_k}\mathbb{E}\left[\|G(x_k,\xi_{k,1})\|^2\middle| \mathcal{F}_k \right] +\frac{1}{m_k^2}\sum_{i\neq j }\mathbb{E}\left[\|G(x_k,\xi_{k,i}) \|\|G(x_k,\xi_{k,j})\|\middle| \mathcal{F}_k \right],
\end{split}
\end{align}
where the first inequality follows from the triangle inequality.
In addition, by using the Cauchy-Schwartz inequality, the second term of the right hand side of \eqref{eq:lem_4_1} can be evaluated as
\begin{align}
    \frac{1}{m_k^2}\sum_{i\neq j }\mathbb{E}\left[\|G(x_k,\xi_{k,i}) \|\|G(x_k,\xi_{k,j})\|\middle| \mathcal{F}_k \right]
    &\le \frac{1}{m_k^2}\sum_{i\neq j }\mathbb{E}\left[\|G(x_k,\xi_{k,i})\|^2\middle| \mathcal{F}_k \right]^{\frac{1}{2}} \mathbb{E}\left[\|G(x_k,\xi_{k,j})\|^2\middle| \mathcal{F}_k \right]^{\frac{1}{2}}\\
    & = \left(1 - \frac{1}{m_k}\right)\mathbb{E}\left[\|G(x_k,\xi_{k,1})\|^2\middle| \mathcal{F}_k \right].
\end{align}
Thus, we obtain from the inequality \eqref{eq:lem_4_1} and the above that
\begin{align}\label{eq:lem_4_2}
\begin{split}
    \mathbb{E}\left[\|G_k\|^2 \middle| \mathcal{F}_k\right] 
  &\le \frac{1}{m_k}\mathbb{E}\left[\|G(x_k,\xi_{k,1})\|^2\middle| \mathcal{F}_k \right] + \left(1 - \frac{1}{m_k}\right)\mathbb{E}\left[\|G(x_k,\xi_{k,1})\|^2\middle| \mathcal{F}_k \right] \\
  & = \mathbb{E}\left[\|G(x_k,\xi_{k,1})\|^2\middle| \mathcal{F}_k \right] \le \sigma^2,
\end{split}
\end{align}
where the last inequality follows from Assumption \ref{assumption_alg} (b).
Combining the inequalities \eqref{eq:lem_4} and \eqref{eq:lem_4_2} yields 
\begin{align}
    \mathbb{E}[\|G_k\|^{1+\nu}] \le  \sigma^{1+\nu},
\end{align}
which is the desired result.
\end{proof}

\begin{lemma}\label{proj_grad_bound}
    Suppose that Assumptions \ref{holder_con} and \ref{assumption_alg} hold and $\{x_k\}$ is generated by Algorithm \ref{SPG} without termination.
    Then, it holds that
    \begin{align}
        \sum_{k=0}^{N-1} \gamma_k \mathbb{E}\left[\| \tilde{g}_{k}\|^2\right] \le \Delta + \sigma^2\sum_{k= 0}^{N-1} \left(\frac{\gamma_k}{m_k}\right)+\frac{L\sigma^{1+\nu}}{1+\nu} \sum_{k= 0}^{N-1} \gamma_k^{1+\nu},
    \end{align}
    for any $N \in \mathbb{N}$, where $\Delta \coloneqq f(x_0)-f^*$.
\end{lemma}

\begin{proof}
    Since $\nabla f$ is H\"older continuous, we obtain from Lemma \ref{hol_leq} and \eqref{eq:gpg} that
\begin{align}
    f(x_{k+1}) &\le f(x_k) + \innerprod<\nabla f(x_k),x_{k+1}-x_k > + \frac{L}{1+\nu}\|x_{k+1} -x_k \|^{1+\nu}\\
    &\le f(x_k) -\gamma_k  \innerprod<\nabla f(x_k),\tilde{g}_{k} > + \frac{L \gamma_k^{1+\nu} }{1+\nu} \|\tilde{g}_{k}\| ^{1+\nu}\\
    &=f(x_k) -  \gamma_k  \innerprod<G_k,\tilde{g}_{k} > - \gamma_k  \innerprod<\nabla f(x_k)-G_k,\tilde{g}_{k} >+ \frac{L \gamma_k^{1+\nu} }{1+\nu} \|\tilde{g}_{k}\| ^{1+\nu}.
\end{align}
By using Lemma \ref{lemma_1} with $x=x_k,~\gamma=\gamma_k,~g= G_k$, we have
\begin{align}
        f(x_{k+1}) &\le f(x_k) -\gamma_k \| \tilde{g}_{k}\|^2 - \gamma_k \left\innerprod<\nabla f(x_k)-G_k,\tilde{g}_{k} \right> + \frac{L \gamma_k^{1+\nu} }{1+\nu} \|\tilde{g}_{k}\| ^{1+\nu}\\
        &=  f(x_k) -\gamma_k \| \tilde{g}_{k}\|^2 - \gamma_k \left\innerprod<\nabla f(x_k)-G_k,g_{k} \right> -\gamma_k\left\innerprod<\nabla f(x_k)-G_k,\tilde{g}_{k}-g_{k} \right>+ \frac{L \gamma_k^{1+\nu} }{1+\nu} \|\tilde{g}_{k}\| ^{1+\nu},
\end{align}
where $g_k \coloneqq P_{\mathcal{X}}(x_k,\nabla f(x_k),\gamma_k)$. Furthermore, it follows from $P_{\mathcal{X}}(x_k,0,\gamma_k) = 0$ that
\begin{align}
       f(x_{k+1}) &\le f(x_k) -\gamma_k \| \tilde{g}_{k}\|^2 - \gamma_k \left\innerprod<\nabla f(x_k)-G_k,g_{k} \right> + \gamma_k\left\|\nabla f(x_k)-G_k\| \|\tilde{g}_{k}-g_{k} \right\|\\
       &\quad + \frac{L \gamma_k^{1+\nu} }{1+\nu} \|\tilde{g}_{k}-P_{\mathcal{X}}(x_k,0,\gamma_k)\| ^{1+\nu}\\
       &\le f(x_k) -\gamma_k \| \tilde{g}_{k}\|^2 - \gamma_k \left\innerprod<\nabla f(x_k)-G_k,g_{k} \right> + \gamma_k\|\nabla f(x_k)-G_k\|^2 + \frac{L \gamma_k^{1+\nu} }{1+\nu} \|G_k\| ^{1+\nu}\label{eq:leq_6},
\end{align}
where the first inequality follows from the Cauchy-Schwartz inequality and the second inequality follows from Lemma \ref{lemma_2}.
Summing up the above inequality for $k=0,1,\ldots,N-1$, we obtain
\begin{align}
     &\sum_{k=0}^{N-1} \gamma_k \| \tilde{g}_{k}\|^2\\
     &\le f(x_0) - f(x_{N}) - \sum_{k=0}^{N-1} \gamma_k \left\innerprod<\nabla f(x_k)-G_k,g_{k} \right>+\sum_{k=0}^{N-1} \gamma_k\|\nabla f(x_k)-G_k\|^2+\sum_{k=0}^{N-1} \frac{L \gamma_k^{1+\nu} }{1+\nu} \|G_k\| ^{1+\nu}\\
    &\le \Delta - \sum_{k=0}^{N-1} \gamma_k \left\innerprod<\nabla f(x_k)-G_k,g_{k} \right>+\sum_{k=0}^{N-1} \gamma_k\|\nabla f(x_k)-G_k\|^2+\sum_{k=0}^{N-1} \frac{L \gamma_k^{1+\nu} }{1+\nu} \|G_k\| ^{1+\nu}\label{eq:leq_7}.
\end{align}
Therefore, taking the expectations of both sides, we have
\begin{align}
    \sum_{k=0}^{N-1} \gamma_k \mathbb{E}\left[ \| \tilde{g}_{k}\|^2\right] &\le \Delta -\sum_{k=0}^{N-1} \gamma_k  \mathbb{E}[ \innerprod<\nabla f(x_k)-G_k,g_{k} >]+\sum_{k=0}^{N-1} \gamma_k \mathbb{E}\left[\|\nabla f(x_k)-G_k\|^2 \right]
    \\
    &\quad+\sum_{k=0}^{N-1} \frac{L \gamma_k^{1+\nu} }{1+\nu}\mathbb{E}\left[ \|G_k\| ^{1+\nu}\right]\label{eq:leq_8}.
\end{align}
Since
\begin{align}
    \mathbb{E}[ \innerprod<\nabla f(x_k)-G_k,g_{k} >] = \mathbb{E}\left[\mathbb{E}[ \innerprod<\nabla f(x_k)-G_k,g_{k} >|\mathcal{F}_k]\right]= \mathbb{E}\left[ \left\innerprod<\mathbb{E}[\nabla f(x_k)-G_k|\mathcal{F}_k],g_{k} \right>\right] = 0
\end{align}
from Assumption \ref{assumption_alg} (a), we obtain from Lemma \ref{lemma_4} that
\begin{align}
    \sum_{k=0}^{N-1} \gamma_k \mathbb{E}\left[\| \tilde{g}_{k}\|^2\right] &\le \Delta +\sum_{k=0}^{N-1}  \gamma_k \mathbb{E}\left[\|\nabla f(x_k)-G_k\|^2 \right]
    +\sum_{k=0}^{N-1} \frac{L \gamma_k^{1+\nu} }{1+\nu}\mathbb{E}\left[ \|G_k\| ^{1+\nu}\right]\\
    & \le \Delta +\sigma^2 \sum_{k=0}^{N-1} \left(\frac{\gamma_k}{m_k}\right)
    +\frac{L \sigma^{1+\nu} }{1+\nu} \sum_{k=0}^{N-1} \gamma_k^{1+\nu} \label{eq:leq_9},
\end{align}
which is the desired result.
\end{proof}

The following theorem provides an upper bound of $\min_{k=0,\ldots,N-1}\mathbb{E}\left[\dist(0,\partial h(x_{k+1})) \right]$.
Convergence is ensured by selecting appropriate stepsizes and batch sizes.

\begin{theorem}\label{grad_bound}
Suppose that Assumptions \ref{holder_con} and \ref{assumption_alg} hold. Let $\{x_k\}$ be generated by Algorithm \ref{SPG} without termination, and there exists $\bar{\gamma} >0$ such that $\gamma_k \le \bar{\gamma}$. Then, for any iteration $N \in \mathbb{N}$,
\begin{enumerate}[(a)]
    \item if $\nu\in(0,1)$, we have
    \begin{align}
    &\min_{k=0,\ldots,N-1}\mathbb{E}\left[\dist(0,\partial h(x_{k+1})) \right]\\
    &\le \frac{ (1+L\nu)\sqrt{N} \sqrt{ \Delta  + \sigma^2 \sum_{k = 0}^{N-1} \left( \frac{\gamma_k}{m_k}\right)+\frac{L \sigma^{1+\nu} }{1+\nu} \sum_{k = 0}^{N-1} \gamma_k^{1+\nu}}+ L(1-\nu) \sum_{k = 0}^{N-1} \gamma_k^{\frac{\nu+1}{2(1-\nu)}} + \sigma \sum_{k = 0}^{N-1} \sqrt{\frac{\gamma_k}{m_k}} }{\sum_{k = 0}^{N-1} \sqrt{\gamma_k}},   
    \end{align}
    
    \item if $\nu = 1$, we have
    \begin{align}
    &\min_{k=0,\ldots,N-1}\mathbb{E}\left[\dist(0,\partial h(x_{k+1})) \right]
    \le \frac{ (1+L \bar{\gamma})\sqrt{N} \sqrt{ \Delta  + \sigma^2 \sum_{k = 0}^{N-1} \left( \frac{\gamma_k}{m_k}\right)+\frac{L \sigma^{2} }{2} \sum_{k = 0}^{N-1} \gamma_k^{2}}+\sigma \sum_{k = 0}^{N-1} \sqrt{\frac{\gamma_k}{m_k}}}{\sum_{k = 0}^{N-1} \sqrt{\gamma_k}},
    \end{align}
\end{enumerate}
 where $\Delta = f(x_0)-f^*$.
\end{theorem}
 
\begin{proof}
Noting that $\delta_{\mathcal{X}}(x)$ remains unchanged when multiplied by a constant, it holds that
\begin{align}
    x_{k+1} = \proj_{\mathcal{X}}(x_k -\gamma_k G_k) =\argmin_{y \in \mathcal{X}} \left\|x_k-\gamma_k G_k-y\right\|^2 =\argmin_{y \in \mathbb{R}^d} \left\{\left\|x_k-\gamma_k G_k-y\right\|^2 +2\gamma_k\delta_{\mathcal{X}}(y)\right\}.
\end{align}
In other words, $x_{k+1}$ is a minimizer of $\left\|x_k-\gamma_k G_k-y\right\|^2 +2\gamma_k\delta_{\mathcal{X}}(y)$, and hence we have $\frac{x_k -x_{k+1}}{\gamma_k} - G_k \in  \partial \delta_{\mathcal{X}}(x_{k+1})$ and $\frac{x_k -x_{k+1}}{\gamma_k} + \nabla f(x_{k+1}) - G_k \in \nabla f(x_{k+1}) + \partial \delta_{\mathcal{X}}(x_{k+1})$.
Thus, it holds that
\begin{align}
    \dist(0,\partial h(x_{k+1}))
    &\le \left\| \frac{x_k -x_{k+1}}{\gamma_k} + \nabla f(x_{k+1}) - G_k\right\|\\
    & = \left\| \frac{x_k -x_{k+1}}{\gamma_k} + \nabla f(x_{k+1}) - \nabla f(x_{k}) + \nabla f(x_{k}) - G_k \right\|\\
    & \le \left\| \frac{x_k -x_{k+1}}{\gamma_k}\right\| + \|\nabla f(x_{k+1}) - \nabla f(x_{k})\| + \|\nabla f(x_{k}) - G_k\|\\
    & \le \| \tilde{g}_{k} \| + L \|x_{k+1} - x_k\|^{\nu} + \|\nabla f(x_{k}) - G_k\| \\
    & = \| \tilde{g}_{k} \| + L \gamma_k^{\nu} \|\tilde{g}_{k}\|^{\nu} + \|\nabla f(x_{k}) - G_k\|,
\end{align}
where the last inequality follows from Assumption \ref{holder_con}.
Since it follows from the concavity of $L\gamma_k^{\nu} (\cdot)^\nu$ that
\begin{align}
    L \gamma_k^{\nu} \|\tilde{g}_{k}\|^{\nu} \le  L \gamma_k^{\nu} c_k^{\nu} + L\nu \gamma_k^{\nu} c_k^{\nu-1}(  \| \tilde{g}_{k} \| - c_k)
\end{align}
for any sequence $\{c_k\}\subset\mathbb{R}_{++}$, we have
\begin{align}
     \dist(0,\partial h(x_{k+1})) &\le \| \tilde{g}_{k} \| +L \gamma_k^{\nu} c_k^{\nu}+ L\nu \gamma_k^{\nu} c_k^{\nu-1}(  \| \tilde{g}_{k} \| - c_k) + \|\nabla f(x_{k}) - G_k\|\\
     & = (1+ L \nu \gamma_k^{\nu} c_k^{\nu-1}) \| \tilde{g}_{k} \|   + L(1-\nu)c_k^{\nu} \gamma_k^{\nu} + \|\nabla f(x_{k}) - G_k\|.
\end{align}
Taking the expectation of both sides yields
\begin{align}
    \mathbb{E}[ \dist(0,\partial h(x_{k+1})) ]& = (1+ L \nu \gamma_k^{\nu} c_k^{\nu-1}) \mathbb{E}[ \| \tilde{g}_{k} \|]  + L(1-\nu)c_k^{\nu} \gamma_k^{\nu} + \mathbb{E}[\|\nabla f(x_{k}) - G_k\|]\\
    &\le  (1+ L \nu \gamma_k^{\nu} c_k^{\nu-1}) \mathbb{E}[ \| \tilde{g}_{k} \|^2]^{\frac{1}{2}}   + L(1-\nu)c_k^{\nu} \gamma_k^{\nu} + \mathbb{E}[\|\nabla f(x_{k}) - G_k\|^2]^{\frac{1}{2}}\\
    &\le (1+ L \nu \gamma_k^{\nu} c_k^{\nu-1}) \mathbb{E}[ \| \tilde{g}_{k} \|^2]^{\frac{1}{2}} + L(1-\nu)c_k^{\nu} \gamma_k^{\nu} + \frac{\sigma}{\sqrt{m_k}},
\end{align}
where the first inequality follows from the Cauchy-Schwartz inequality, and the last inequality from Lemma \ref{lemma_4} (a). 
By multiplying both sides of the above inequality by $\sqrt{\gamma_k}$ and summing up it for $k=0,1,\ldots,N-1$,  we obtain from the Cauchy-Schwartz inequality that
\begin{align}
    \sum_{k=0}^{N-1} \sqrt{\gamma_k} \mathbb{E}[ \dist(0,\partial h(x_{k+1})) ] &\le  \sum_{k=0}^{N-1} (1+ L \nu \gamma_k^{\nu} c_k^{\nu-1}) \sqrt{\gamma_k \mathbb{E}[ \| \tilde{g}_{k} \|^2] } +  L(1-\nu) \sum_{k=0}^{N-1}c_k^{\nu} \gamma_k^{\nu+\frac{1}{2}} + \sigma \sum_{k=0}^{N-1}  \sqrt{\frac{\gamma_k}{m_k}}\\
    & \le \sqrt{ \sum_{k=0}^{N-1} (1+ L \nu \gamma_k^{\nu} c_k^{\nu-1})^2 } \sqrt{\sum_{k=0}^{N-1}  \gamma_k \mathbb{E}[ \| \tilde{g}_{k} \|^2] } +  L(1-\nu) \sum_{k=0}^{N-1}c_k^{\nu} \gamma_k^{\nu+\frac{1}{2}} \\ 
    &\quad + \sigma \sum_{k=0}^{N-1}  \sqrt{\frac{\gamma_k}{m_k}}.
\end{align}
We obtain from the above inequality and Lemma \ref{proj_grad_bound} that
\begin{align}\label{eq:leq_11}
\begin{split}
   &\min_{k=0,\ldots,N-1}\mathbb{E}[ \dist(0,\partial h(x_{k+1})) ] \\
   &\le \frac{\sqrt{ \sum_{k=0}^{N-1} (1+ L \nu \gamma_k^{\nu} c_k^{\nu-1})^2 } \sqrt{\sum_{k=0}^{N-1}  \gamma_k \mathbb{E}[ \| \tilde{g}_{k} \|^2] }+  L(1-\nu) \sum_{k=0}^{N-1}c_k^{\nu} \gamma_k^{\nu+\frac{1}{2}} + \sigma \sum_{k=0}^{N-1}  \sqrt{\frac{\gamma_k}{m_k}}}{\sum_{k=0}^{N-1} \sqrt{\gamma_k}}\\
   &\le \frac{\sqrt{ \sum_{k=0}^{N-1} (1+ L \nu \gamma_k^{\nu} c_k^{\nu-1})^2 } \sqrt{\Delta  + \sigma^2 \sum_{k = 0}^{N-1} \left( \frac{\gamma_k}{m_k}\right)+\frac{L \sigma^{1+\nu} }{1+\nu} \sum_{k = 0}^{N-1} \gamma_k^{1+\nu} }+  L(1-\nu) \sum_{k=0}^{N-1}c_k^{\nu} \gamma_k^{\nu+\frac{1}{2}}}{\sum_{k=0}^{N-1} \sqrt{\gamma_k}}\\
   &\quad+ \frac{\sigma \sum_{k=0}^{N-1}  \sqrt{\frac{\gamma_k}{m_k}}}{\sum_{k=0}^{N-1} \sqrt{\gamma_k}},
\end{split}
\end{align}
if $\nu \in (0,1)$, setting $c_k = \gamma_k^{\frac{\nu}{1-\nu}}$, it holds that
\begin{align}
&\min_{k=0,\ldots,N-1}\mathbb{E}\left[\dist(0,\partial h(x_{k+1})) \right]\\
&\le \frac{ (1+L\nu)\sqrt{N} \sqrt{ \Delta  + \sigma^2 \sum_{k = 0}^{N-1} \left( \frac{\gamma_k}{m_k}\right)+\frac{L \sigma^{1+\nu} }{1+\nu} \sum_{k = 0}^{N-1} \gamma_k^{1+\nu}}+L(1-\nu) \sum_{k = 0}^{N-1} \gamma_k^{\frac{\nu+1}{2(1-\nu)}} + \sigma \sum_{k = 0}^{N-1} \sqrt{\frac{\gamma_k}{m_k}}  }{\sum_{k = 0}^{N-1} \sqrt{\gamma_k}}.
\end{align}
On the other hand, if $\nu=1$, the inequality \eqref{eq:leq_11} can be evaluated as
\begin{align}
   &\min_{k=0,\ldots,N-1}\mathbb{E}[ \dist(0,\partial h(x_{k+1})) ]\\
   &\le \frac{\sqrt{ \sum_{k=0}^{N-1} (1+ L\bar{\gamma} )^2 } \sqrt{ \Delta  + \sigma^2 \sum_{k = 0}^{N-1} \left( \frac{\gamma_k}{m_k}\right)+\frac{L \sigma^{2} }{2} \sum_{k = 0}^{N-1} \gamma_k^{2}}+ \sigma \sum_{k=0}^{N-1}  \sqrt{\frac{\gamma_k}{m_k}}}{\sum_{k=0}^{N-1} \sqrt{\gamma_k}}\\
   &\le \frac{(1+L \bar{\gamma}  )\sqrt{ N } \sqrt{ \Delta  + \sigma^2 \sum_{k = 0}^{N-1} \left( \frac{\gamma_k}{m_k}\right)+\frac{L \sigma^{2} }{2} \sum_{k = 0}^{N-1} \gamma_k^{2}}+ \sigma \sum_{k=0}^{N-1}  \sqrt{\frac{\gamma_k}{m_k}}}{\sum_{k=0}^{N-1} \sqrt{\gamma_k}}.
\end{align}
we have the desired result.
\end{proof}

From Theorem \ref{grad_bound}, we derive an explicit convergence rate and the stochastic gradient call complexity for finding an $\varepsilon$-stationary solution $x$ that satisfies
\begin{align}
     \mathbb{E}\left[\dist(0,\partial h(x)) \right] \le \varepsilon,
\end{align}
where $\varepsilon > 0$.

\begin{corollary}\label{convergence_result}
    Suppose that Assumptions \ref{holder_con} and \ref{assumption_alg} hold and $\{x_k\}$ is generated by Algorithm \ref{SPG} without termination. Let stepsizes and batch sizes be set to 
    \begin{align}
    &\gamma_k = \frac{\tilde{\gamma}}{(1+k)^{\beta_1}},\label{eq:stepsize}\\
    & m_k =  \left\lceil \tilde{m} \left(1+k\right)^{\beta_2}\right\rceil\label{eq:batchsize}
    \end{align}
    with $\beta_1 = \frac{1}{1+\nu},~ \beta_2 = \frac{\nu}{1+\nu},~ \tilde{\gamma}>0,~ \tilde{m}>0$, respectively.
    Then, we have
    \begin{equation}
        \min_{k=0,\ldots,N-1}\mathbb{E}\left[\dist(0,\partial h(x_{k+1})) \right] = \mathcal{O}\left(\sqrt{\frac{\log N}{N^{\frac{\nu}{1+\nu}}}}\right).
    \end{equation}
    Moreover, the stochastic gradient call complexity to obtain an $\varepsilon$-stationary solution is of the order $\varepsilon^{-\frac{4\nu+2}{\nu-\delta(1+\nu)}}$ for any $\delta \in (0, \frac{\nu}{1+\nu})$.
\end{corollary}

\begin{proof}
If $\nu \in (0,1)$, replacing the positive constants by $C_1,\ldots,C_5$ for the upper bound in Theorem \ref{grad_bound} (a) yields
    \begin{align}
        &\min_{k=0,\ldots,N-1}\mathbb{E}\left[\dist(0,\partial h(x_{k+1})) \right]\\
        &\le \frac{\sqrt{N} \sqrt{ C_1  + C_2\sum_{k = 0}^{N-1} \left( \frac{\gamma_k}{m_k}\right)+C_3\sum_{k = 0}^{N-1} \gamma_k^{1+\nu}}+C_4\sum_{k = 0}^{N-1}\gamma_k^{\frac{\nu+1}{2(1-\nu)}} + C_5\sum_{k = 0}^{N-1} \sqrt{\frac{\gamma_k}{m_k}}  }{\sum_{k = 0}^{N-1} \sqrt{\gamma_k}}\\
        &\le \frac{\sqrt{N} \sqrt{ C_1  + \left(\frac{C_2 \tilde{\gamma}}{\tilde{m}} + C_3\tilde{\gamma}^{1+\nu} \right)\sum_{k = 0}^{N-1} \frac{1}{1+k}}+C_4 \tilde{\gamma}^{\frac{\nu+1}{2(1-\nu)}} \sum_{k = 0}^{N-1} \frac{1}{(1+k)^{\frac{1}{2(1-\nu)}}} + C_5 \sqrt{\frac{\tilde{\gamma}}{\tilde{m}}}\sum_{k = 0}^{N-1} \frac{1}{\sqrt{1+k}}  }{\sqrt{\tilde{\gamma}}\sum_{k = 0}^{N-1} \frac{1}{(1+k)^{\frac{1}{2(1+\nu)}}}}.
    \end{align}
Since it follows that
    \begin{align}
        &\sum_{k=0}^{N-1} \frac{1}{(1+k)^\frac{1}{2(1+\nu)}} \geq \int_{0}^{N} \frac{1}{(1+t)^\frac{1}{2(1+\nu)}} dt = \frac{2(1+\nu)}{1+2\nu}\left\{(1+N)^{\frac{1+2\nu}{2(1+\nu)}} -1 \right\},\\
        &\sum_{k=0}^{N-1}\frac{1}{(1+k)} \le 1+\int_{0}^{N-1}\frac{1}{1+t} dt = 1+ \log N,\\
        & \sum_{k=0}^{N-1}\frac{1}{(1+k)^{\frac{1}{2(1-\nu)}}} \le 1+\int_{0}^{N-1} \frac{1}{(1+t)^{\frac{1}{2(1-\nu)}}} dt = \left\{
        \begin{alignedat}{3}
            & 1+\log N,    &\quad&   \nu = \frac{1}{2},   \\
            & 1+\frac{2(1-\nu)}{1-2\nu}\left\{N^{\frac{1-2\nu}{2(1-\nu)}} -1 \right\}, &\quad&  \nu \neq \frac{1}{2}
        \end{alignedat} \right. =\mathcal{O}(N^\frac{1}{2}),\\
        &\sum_{k=0}^{N-1}\frac{1}{\sqrt{1+k}} \le  1+\int_{0}^{N-1}\frac{1}{\sqrt{1+t}} dt =
        1+2(N^\frac{1}{2}-1),
    \end{align}
we have
\begin{align}\label{eq:con_rate}
    \min_{k=0,\ldots,N-1}\mathbb{E}\left[\dist(0,\partial h(x_{k+1})) \right] = \mathcal{O}\left( \sqrt{\frac{\log N}{ N^{\frac{\nu}{1+\nu}}}}\right).
\end{align}
Similarly, when $\nu = 1$, we achieve the convergence rate from the upper bound in Theorem \ref{grad_bound} (b). 

Since $\log N=\mathcal{O}(N^\delta)$ for any $\delta \in \left(0,\frac{\nu}{1+\nu}\right)$, a number of iterations of the order $\varepsilon^{-\frac{2(1+\nu)}{\nu -\delta(1+\nu) }}$ is required to obtain an $\varepsilon$-stationary solution.
In addition, the total number of stochastic gradient calls up to iteration $N-1$ can be evaluated as
\begin{align}
\sum_{k=0}^{N-1} m_k =  \sum_{k=0}^{N-1} \left\lceil \tilde{m} \left(1+k\right)^{\frac{\nu}{1+\nu}}\right\rceil
\le \tilde{m}\left( N + 1 + \sum_{k=0}^{N-1}\left(1+k\right)^{\frac{\nu}{1+\nu}} \right)
&\le \tilde{m}\left( N + 1 + \int_{0}^{N-1}\left(1+t\right)^{\frac{\nu}{1+\nu}} dt \right) \\
&= \mathcal{O}\left( N ^\frac{2\nu+1}{1+\nu}\right).
\end{align}
Consequently, the stochastic gradient call complexity is of the order $\varepsilon^{-\frac{4\nu+2}{\nu -\delta(1+\nu)}}$ for any $\nu \in (0,1]$ and $\delta \in \left(0,\frac{\nu}{1+\nu}\right)$.
\end{proof}

From Corollary \ref{convergence_result}, we need only to know H\"older exponent $\nu$ and do not need to specify $L$ to execute Algorithm \ref{SPG} with guaranteed convergence results.
Thus, Algorithm \ref{SPG} is a straightforward approach that only requires the determination of appropriate diminishing stepsizes and increasing batch sizes based on $\nu$, under suitable assumptions and the availability of stochastic gradients.

\section{Numerical experiments}\label{sec:experiment}
Two numerical experiments to demonstrate the effectiveness of our proposed methods are conducted.
In Subsection \ref{experiment_1}, based on the convergence results of Corollary \ref{convergence_result}, we confirm that our proposed methods reduce the ruin probability.
In Subsection \ref{experiment_2}, we statistically compare the adjustment coefficient approach \citep{hald2004maximisation} with our proposed methods.
We consider both proportional reinsurance and investment as strategies in Subsection \ref{experiment_1}.
On the other hand, only proportional reinsurance is considered in Subsection \ref{experiment_2} because the adjustment coefficient approach can only handle reinsurance.
Algorithms \ref{SPG} were terminated when the number of iterations reached $N_{\max}$.
We used the function $w(t) = t^{\frac{1}{8}}(T-t)^{\frac{1}{8}}$ and $T=5$.
For the surplus model \eqref{eq:surplus_model}, we also set the premium rate $c_1$ and the reinsurance premium rate $(1-b)c_2$ to $\lambda (1+\theta)\mathbb{E}[X_1]$ and $\lambda (1+\zeta)(1-b)\mathbb{E}[X_1]$, respectively, based on the expected value principle, where $\theta>0$ is the safety loading of the insurer and $\zeta>0$ is the safety loading of the reinsurer.

\subsection{Convergence behavior of ruin probability}\label{experiment_1}
In this subsection, we choose $\lambda = 40,~r = 0.05,~u=640,~\theta = 0.08,~\underline{b} = 0.19$, and $\zeta = 0.32$ for the surplus model \eqref{eq:surplus_model}.
Assume that the claim size follows a gamma distribution $\Gamma(5,3)$.
We consider three patterns of the number of assets in the investment: $n = 11,101,1001$.
Suppose that $S_t^{(1)}\equiv 1$ and $S_t^{(j)}$ is defined as
\begin{align}\label{brownian_motion}
    S_t^{(j)} = \exp\left\{ \left(\mu_j - \frac{\sigma_j^2}{2}\right)t + \sigma_j W_t^{(j)} \right\},
\end{align}
for $j=2,\ldots,n$, where $\{W_t^{(j)}\}_{j=2}^n$ are i.i.d. standard Brownian motions with $W_0^{(j)}=0$, $\mu_j\in\mathbb{R}$, and $\sigma_j>0$.
The parameters $\{\mu_j\}_{j=2}^n$ and $\{\sigma_j\}_{j=2}^n$ follow from i.i.d. uniform distributions $U(-0.05,0.1)$ and $U(0.005,0.01)$, respectively.
The stepsizes and batch sizes are set as in \eqref{eq:stepsize} and \eqref{eq:batchsize} with $(\beta_1,\beta_2) = (0.67,0.33),(0.7,0.5),(0.7,0.2),(0.9,0.5)$, $\tilde{m}=1$, and $\tilde{\gamma} = 0.1,1,10$.
These step and mini-batch sizes guarantee convergence of the upper bound of the Theorem \ref{grad_bound} to zero, especially $\beta_1 = 0.67,~\beta_2 = 0.33$ is suggested in Corollary \ref{convergence_result}. 
Indeed, the existence of any moments of $S_T^{(j)}$ and $X_i$ in our setting allows the H\"older exponent of $\nabla F$ to be $\frac{33}{67}$.

Figures \ref{fig_conv_11}-\ref{fig_conv_1001} show the convergence behaviors of the minimum ruin probability $\min_{k=0,\ldots,N}F(p_k,b_k)$.
We used the Monte Carlo method with sample size $10,000$ to approximate the ruin probability $F(p_k,b_k)$ at iteration $k$.
From Figures \ref{fig_conv_11}-\ref{fig_conv_1001}, we see that our proposed method succeeded in reducing the ruin probability.
In addition, the ruin probability decreased as the number of investable risk assets or cash assets increased.
This result is intuitive, as having more assets to choose from allows for a better selection of high-performing assets, indicating that our approach works properly.
Observing Figures \ref{fig_conv_11}-\ref{fig_conv_1001} with respect to the initial step size, we see that the objective function decreased more significantly with the larger initial stepsize.

\begin{figure}[H]
\begin{center}
    \includegraphics[width=1\textwidth]{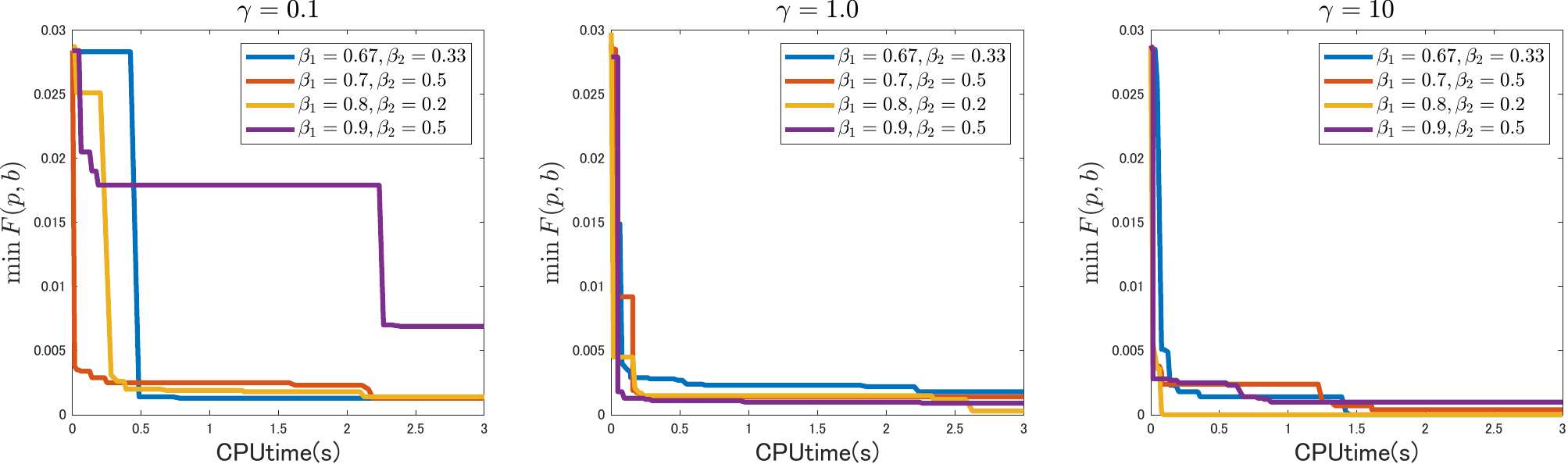}    
\end{center}
\caption{Convergence behaviors of the minimum ruin probability with $n= 11$.}
\label{fig_conv_11}
\end{figure}

\begin{figure}[H]
\begin{center}
    \includegraphics[width=1\textwidth]{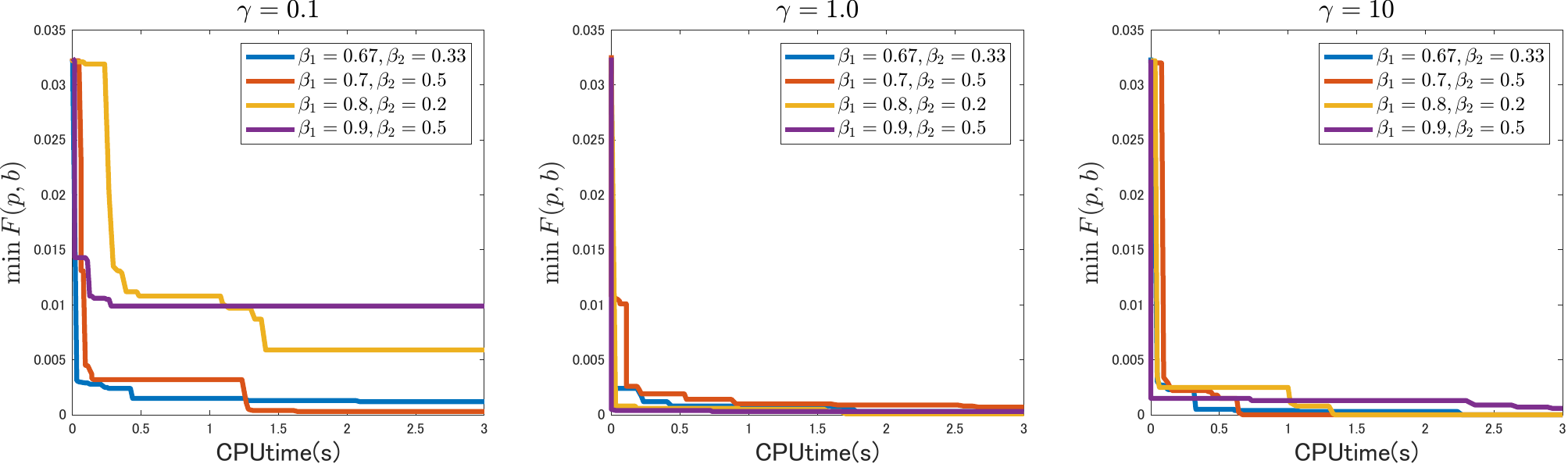}    
\end{center}
\caption{Convergence behaviors of the minimum ruin probability with $n= 101$.}
\label{fig_conv_101}
\end{figure}

\begin{figure}[H]
\begin{center}
    \includegraphics[width=1\textwidth]{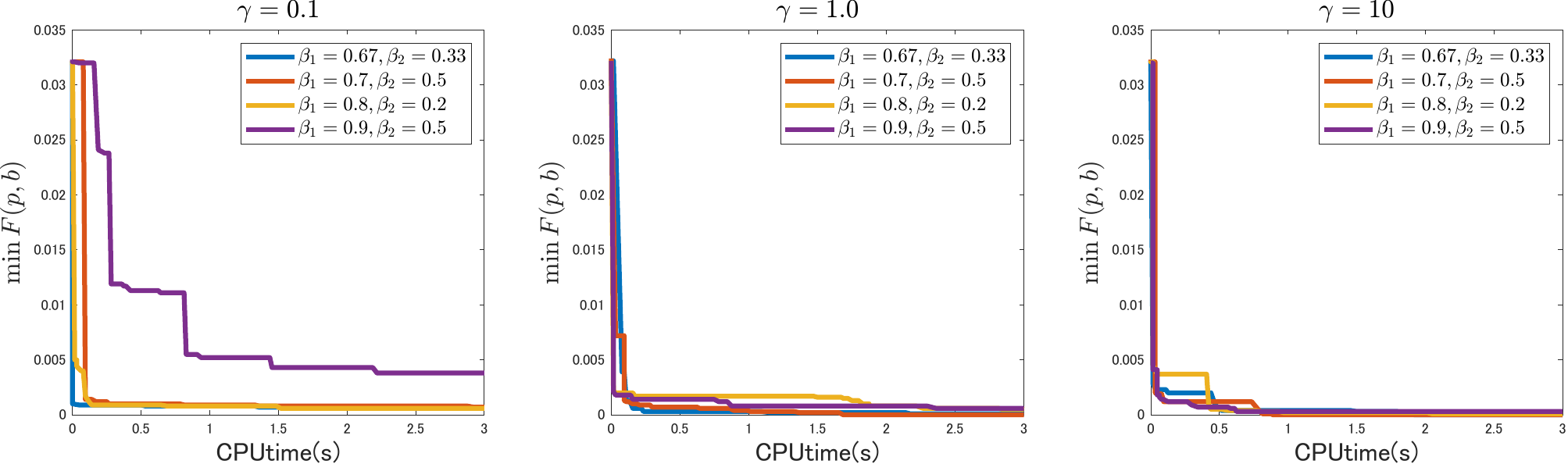}    
\end{center}
\caption{Convergence behaviors of the minimum ruin probability with $n=1001$.}
\label{fig_conv_1001}
\end{figure}

\subsection{Comparison with the adjustment coefficient approach}\label{experiment_2}
In this subsection, we choose  $\lambda = 40,~r = 0.03,~u=200,~\theta = 0.08,~\underline{b} = 0.07$, and $\zeta = 0.15$ for the surplus model \eqref{eq:surplus_model}.
We consider three different gamma distributions $\Gamma(5, 3), \Gamma(10, 3)$, and $\Gamma(15, 3)$ for the claim size.
The stepsizes and batch sizes are set as in \eqref{eq:stepsize} and \eqref{eq:batchsize} with $\beta_1 = 0.67$, $\beta_2 = 0.33$, $\tilde{m}=1$ and $\tilde{\gamma} = 0.01,0.1,1.0$.

To compare the adjustment coefficient approach with our proposed method, we provide an upper bound of $\mathbb{P}(U_T(b) < 0)$. we consider the following surplus process $\tilde{U}_t(b)$ such that $\tilde{U}_t(b) \le U_t(b)$ for any $t\in[0,T]$.
\begin{align}\label{eq:surplus_model_fix_infration}
    \begin{split}
    	\tilde{U}_t(b) =~ & u+c_1t-(1-b)c_2t - \sum_{i=1}^{N_t} be^{rT}  X_i.
    \end{split}
\end{align}
Then, we have
\begin{align}
    \mathbb{P}(U_T(b) < 0) &\le \mathbb{P}\left( \inf_{0\le t\le T} U_t(b) <0 \right) \le \mathbb{P}\left( \inf_{0 \le t\le T} \tilde{U}_t(b) <0 \right) \le \mathbb{P}\left( \inf_{t\geq 0} \tilde{U}_t(b) <0 \right) < e^{-R(b)u},
\end{align}
where the last inequality is the Lundberg inequality and  $R(b)$ is the adjustment coefficient.
From \citep[Example 1]{hald2004maximisation}, the optimal retention level $b^*$ obtained by the
adjustment coefficient approach is expressed as
\begin{align}\label{eq:ad_sol}
    b^* = \min \left(\frac{\alpha(\zeta -\theta)\left(1-(1+\zeta)^{-1/(\alpha+1)}\right)}{\alpha \zeta + (\alpha + 1)(1-(1+\zeta)^{\alpha/(\alpha+1)})}, 1\right),
\end{align}
where $\alpha$ is the shape parameter of the gamma distribution, namely, $\alpha = 5,10,15$.

Figures \ref{fig_box_0.01}-\ref{fig_box_1} show box plots of ruin probabilities, (A) $\mathbb{P}(U_T < 0)$ and (B) $\mathbb{P}(\inf_{0\le t \le T} U_t < 0)$, for solutions obtained by executing our method 50 times.
Additionally, ruin probabilities for the solution \eqref{eq:ad_sol} obtained by the adjustment coefficient approach are shown as a blue dot in Figures \ref{fig_box_0.01}-\ref{fig_box_1}.
To estimate these ruin probabilities, we performed the Monte Carlo method with sample size $10,000$.

In Figure \ref{fig_box_0.01}, the medians of the ruin probabilities obtained using our proposed method and the ruin probabilities obtained by the adjustment coefficient approach are at the same level.
On the other hand, as shown in Figures \ref{fig_box_0.1}-\ref{fig_box_1}, we see that the median of the ruin probabilities obtained using our proposed method appeared to be lower than that of the adjustment coefficient approach.
These results show that it is possible to obtain a solution better than that obtained by the adjustment coefficient approach with high probability by running our algorithm from some initial values.
In addition, observing Figure \ref{fig_box_0.01}-\ref{fig_box_1} with respect to the initial step size, both ruin probabilities decreased more when the initial stepsize was larger.

\begin{figure}[H]
\begin{center}
    \includegraphics[width=0.9\textwidth]{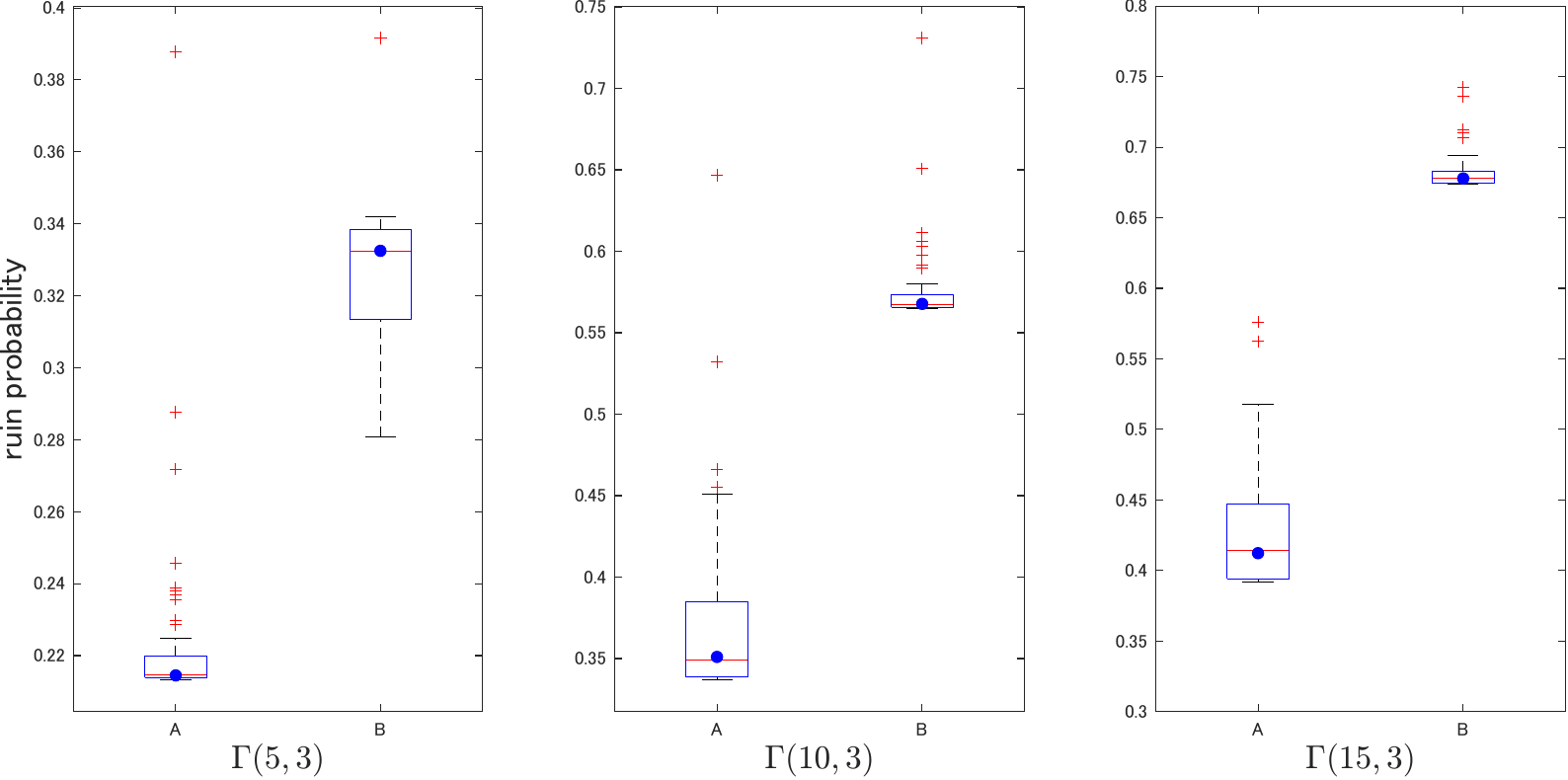}    
\end{center}
\caption{Ruin probabilities of the proposed method and the adjustment coefficient approach with $\tilde{\gamma} = 0.01$.}
\label{fig_box_0.01}
\end{figure}

\begin{figure}[H]
\begin{center}
    \includegraphics[width=0.9\textwidth]{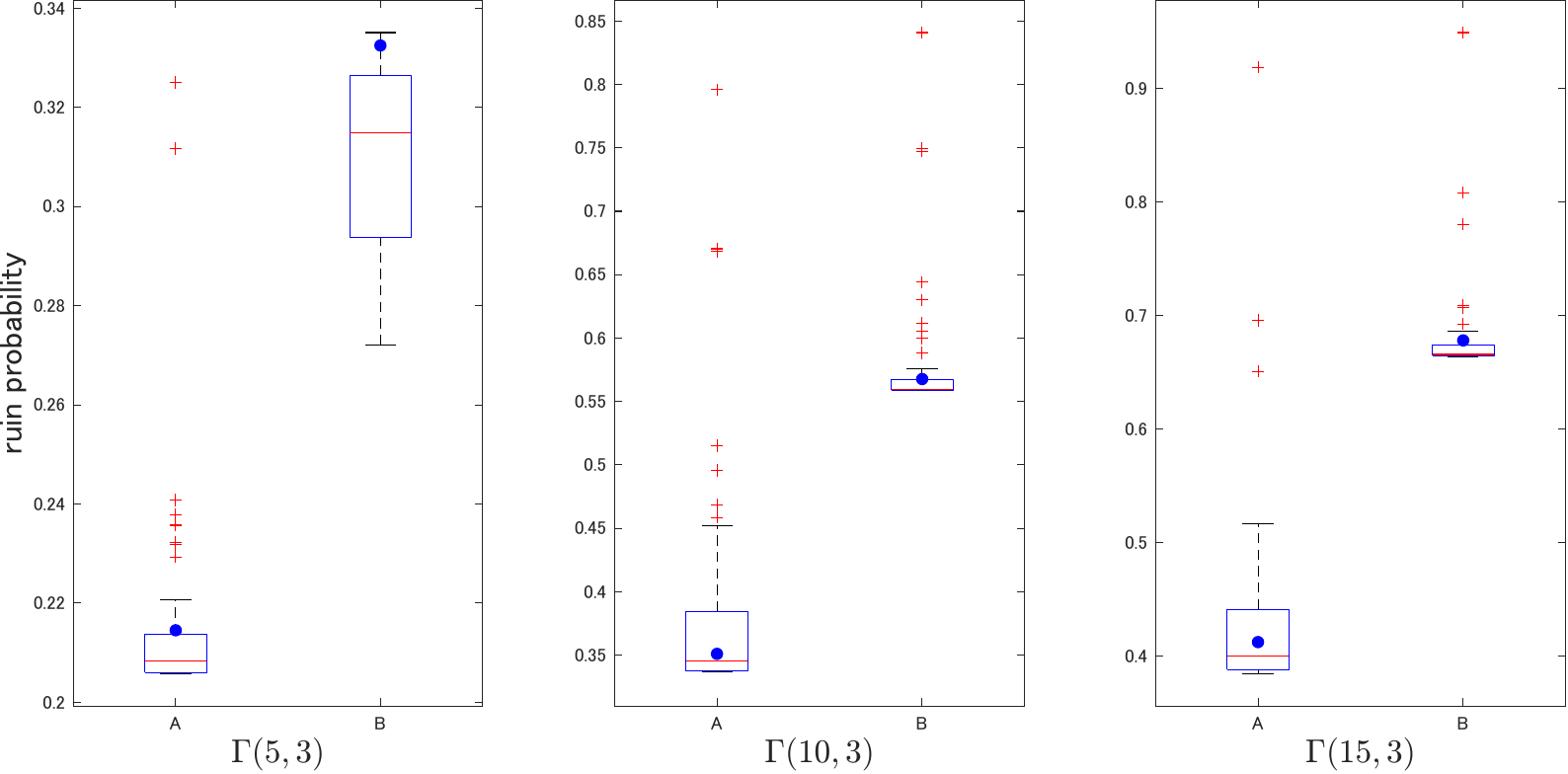}    
\end{center}
\caption{Ruin probabilities of the proposed method and the adjustment coefficient approach with $\tilde{\gamma} = 0.1$.}
\label{fig_box_0.1}
\end{figure}

\begin{figure}[H]
\begin{center}
    \includegraphics[width=0.9\textwidth]{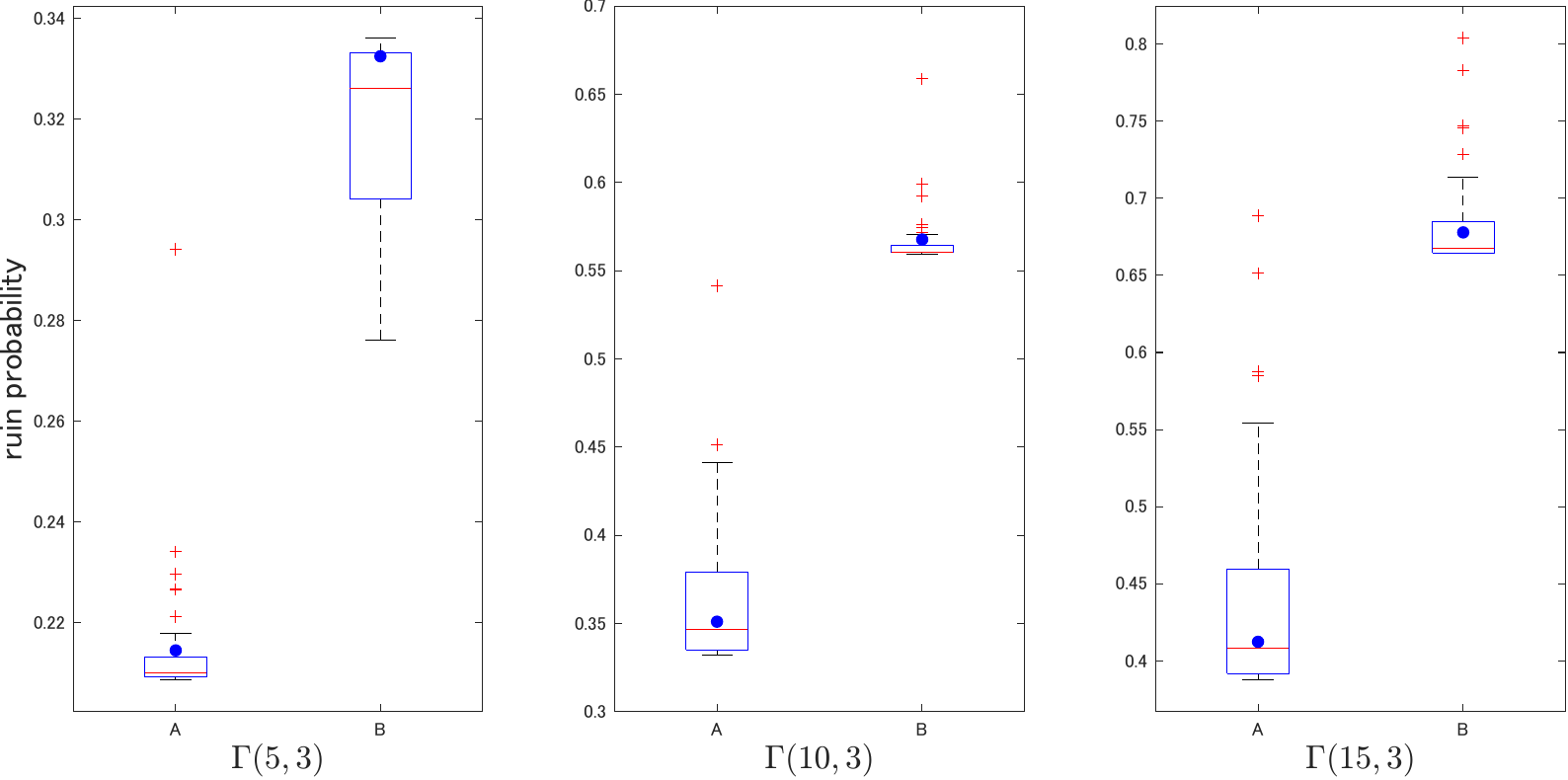}    
\end{center}
\caption{Ruin probabilities of the proposed method and the adjustment coefficient approach with $\tilde{\gamma} = 1.0$.}
\label{fig_box_1}
\end{figure}

\section{Concluding Remarks}\label{sec:conclusion}
This paper has proposed a stochastic projected gradient method combined with Malliavin calculus to find optimal investment and reinsurance strategies.
Firstly, the H\"older continuity of the gradient of the ruin probability has been shown.
Secondly, we have considered a general constrained nonconvex optimization problem under the H\"older condition and have established convergence results of the stochastic projected gradient method.
Finally, our numerical experiments have illustrated the effectiveness of our proposed method.
An important direction for future work is to further develop our approach for the ruin probability \eqref{eq:ruin_prob_finte}, as it serves as a more realistic probability of ruin than the ruin probability \eqref{eq:ruin_prob}.

\bibliography{reference.bib}

\begin{thebibliography}{25}
\providecommand{\natexlab}[1]{#1}
\providecommand{\url}[1]{\texttt{#1}}
\expandafter\ifx\csname urlstyle\endcsname\relax
  \providecommand{\doi}[1]{doi: #1}\else
  \providecommand{\doi}{doi: \begingroup \urlstyle{rm}\Url}\fi

\bibitem[Avikainen(2009)]{avikainen2009irregular}
Rainer Avikainen.
\newblock On irregular functionals of sdes and the euler scheme.
\newblock \emph{Finance and Stochastics}, 13\penalty0 (3):\penalty0 381--401, 2009.

\bibitem[Beck(2017)]{beck2017first}
Amir Beck.
\newblock \emph{First-order methods in optimization}.
\newblock SIAM, 2017.

\bibitem[Centeno(1986)]{CENTENO1986169}
Lourdes Centeno.
\newblock Measuring the effects of reinsurance by the adjustment coefficient.
\newblock \emph{Insurance: Mathematics and Economics}, 5\penalty0 (2):\penalty0 169--182, 1986.
\newblock ISSN 0167-6687.

\bibitem[Daley et~al.(2003)Daley, Vere-Jones, et~al.]{daley2003introduction}
Daryl~J Daley, David Vere-Jones, et~al.
\newblock \emph{An introduction to the theory of point processes: volume I: elementary theory and methods}.
\newblock Springer, 2003.

\bibitem[Ewald(2006)]{ewald2006malliavin}
Christian-Oliver Ewald.
\newblock The malliavin gradient method for the calibration of stochastic dynamical models.
\newblock \emph{Applied mathematics and computation}, 175\penalty0 (2):\penalty0 1332--1352, 2006.

\bibitem[Ewald and Zhang(2006)]{ewald2006new}
Christian-Oliver Ewald and Aihua Zhang.
\newblock A new technique for calibrating stochastic volatility models: the malliavin gradient method.
\newblock \emph{Quantitative Finance}, 6\penalty0 (02):\penalty0 147--158, 2006.

\bibitem[Ghadimi and Lan(2013)]{ghadimi2013stochastic}
Saeed Ghadimi and Guanghui Lan.
\newblock Stochastic first-and zeroth-order methods for nonconvex stochastic programming.
\newblock \emph{SIAM journal on optimization}, 23\penalty0 (4):\penalty0 2341--2368, 2013.

\bibitem[Ghadimi et~al.(2016)Ghadimi, Lan, and Zhang]{ghadimi2016mini}
Saeed Ghadimi, Guanghui Lan, and Hongchao Zhang.
\newblock Mini-batch stochastic approximation methods for nonconvex stochastic composite optimization.
\newblock \emph{Mathematical Programming}, 155\penalty0 (1):\penalty0 267--305, 2016.

\bibitem[Grandell(2012)]{grandell2012aspects}
Jan Grandell.
\newblock \emph{Aspects of risk theory}.
\newblock Springer Science \& Business Media, 2012.

\bibitem[Hald and Schmidli(2004)]{hald2004maximisation}
Morten Hald and Hanspeter Schmidli.
\newblock On the maximisation of the adjustment coefficient under proportional reinsurance.
\newblock \emph{ASTIN Bulletin: The Journal of the IAA}, 34\penalty0 (1):\penalty0 75--83, 2004.

\bibitem[Hu and Yuen(2012)]{hu2012optimal}
Fengqing Hu and Kam~C Yuen.
\newblock Optimal proportional reinsurance under dependent risks.
\newblock \emph{Journal of Systems Science and Complexity}, 25\penalty0 (6):\penalty0 1171--1184, 2012.

\bibitem[Lei et~al.(2019)Lei, Hu, Li, and Tang]{lei2019stochastic}
Yunwen Lei, Ting Hu, Guiying Li, and Ke~Tang.
\newblock Stochastic gradient descent for nonconvex learning without bounded gradient assumptions.
\newblock \emph{IEEE transactions on neural networks and learning systems}, 31\penalty0 (10):\penalty0 4394--4400, 2019.

\bibitem[Liang and Guo(2007)]{liang2007optimal}
Zhibin Liang and Junyi Guo.
\newblock Optimal proportional reinsurance and ruin probability.
\newblock \emph{Stochastic Models}, 23\penalty0 (2):\penalty0 333--350, 2007.

\bibitem[Liang and Guo(2008)]{liang2008upper}
Zhibin Liang and Junyi Guo.
\newblock Upper bound for ruin probabilities under optimal investment and proportional reinsurance.
\newblock \emph{Applied Stochastic Models in Business and Industry}, 24\penalty0 (2):\penalty0 109--128, 2008.

\bibitem[Liang and Guo(2012)]{liang2012optimal}
Zhibin Liang and Junyi Guo.
\newblock Optimal investment and proportional reinsurance in the sparre andersen model.
\newblock \emph{Journal of Systems Science and Complexity}, 25\penalty0 (5):\penalty0 926--941, 2012.

\bibitem[Meng et~al.(2023)Meng, Wei, and Zhou]{meng2023multiple}
Hui Meng, Li~Wei, and Ming Zhou.
\newblock Multiple per-claim reinsurance based on maximizing the lundberg exponent.
\newblock \emph{Insurance: Mathematics and Economics}, 112:\penalty0 33--47, 2023.

\bibitem[Patel and Zhang(2021)]{patel2021stochastic}
Vivak Patel and Shushu Zhang.
\newblock Stochastic gradient descent on nonconvex functions with general noise models.
\newblock \emph{arXiv preprint arXiv:2104.00423}, 2021.

\bibitem[Patel et~al.(2022)Patel, Zhang, and Tian]{patel2022global}
Vivak Patel, Shushu Zhang, and Bowen Tian.
\newblock Global convergence and stability of stochastic gradient descent.
\newblock \emph{Advances in Neural Information Processing Systems}, 35:\penalty0 36014--36025, 2022.

\bibitem[Privault and Wei(2004)]{privault2004malliavin}
Nicolas Privault and Xiao Wei.
\newblock A malliavin calculus approach to sensitivity analysis in insurance.
\newblock \emph{Insurance: Mathematics and Economics}, 35\penalty0 (3):\penalty0 679--690, 2004.

\bibitem[Schmidli(2001)]{schmidli2001optimal}
Hanspeter Schmidli.
\newblock Optimal proportional reinsurance policies in a dynamic setting.
\newblock \emph{Scandinavian Actuarial Journal}, 2001\penalty0 (1):\penalty0 55--68, 2001.

\bibitem[Schmidli(2002)]{schmidli2002minimizing}
Hanspeter Schmidli.
\newblock On minimizing the ruin probability by investment and reinsurance.
\newblock \emph{The Annals of Applied Probability}, 12\penalty0 (3):\penalty0 890--907, 2002.

\bibitem[Taylor(1979)]{taylor1979probability}
Gregory~Clive Taylor.
\newblock Probability of ruin under inflationary conditions or under experience rating.
\newblock \emph{ASTIN Bulletin: The Journal of the IAA}, 10\penalty0 (2):\penalty0 149--162, 1979.

\bibitem[Waters(1983)]{waters1983some}
Howard~R Waters.
\newblock Some mathematical aspects of reinsurance.
\newblock \emph{Insurance: Mathematics and Economics}, 2\penalty0 (1):\penalty0 17--26, 1983.

\bibitem[Yashtini(2016)]{yashtini2016global}
Maryam Yashtini.
\newblock On the global convergence rate of the gradient descent method for functions with h{\"o}lder continuous gradients.
\newblock \emph{Optimization letters}, 10:\penalty0 1361--1370, 2016.

\bibitem[Zhang and Liang(2016)]{zhang2016optimal}
Xuepeng Zhang and Zhibin Liang.
\newblock Optimal layer reinsurance on the maximization of the adjustment coefficient.
\newblock \emph{Numerical Algebra, Control and Optimization}, 6\penalty0 (1):\penalty0 21, 2016.

\end{thebibliography}
\bibliographystyle{plainnat}

\end{document}